\documentclass[prx,a4paper,aps,twocolumn,superscriptaddress,10pt,longbibliography]{revtex4-1}

\usepackage{amsmath,amsthm,graphicx,xcolor,times,xfrac,booktabs,mathtools,enumitem,xr,subfigure,bbm,verbatim,appendix,placeins,comment}
\usepackage[unicode=true,bookmarks=true,bookmarksnumbered=false,bookmarksopen=false,breaklinks=false,pdfborder={0 0 1}, backref=false,colorlinks=true]{hyperref}
\renewcommand*{\url}[1]{\href{#1}{#1}}

\makeatletter
\theoremstyle{plain}
\newtheorem{thm}{\protect\theoremname}
\theoremstyle{plain}
\newtheorem{lem}{\protect\lemmaname}
\theoremstyle{plain}

\theoremstyle{remark}
\newtheorem*{rem*}{\protect\remarkname}
\theoremstyle{plain}

\theoremstyle{plain}
\newtheorem{cor}{\protect\corollaryname}
\theoremstyle{definition}
\newtheorem{defn}{\protect\definitionname}
\theoremstyle{plain}
\newtheorem*{thm*}{\protect\theoremname}
\theoremstyle{plain}
\newtheorem*{lem*}{\protect\lemmaname}
 
\providecommand{\propositionname}{Proposition}
\providecommand{\theoremname}{Theorem}
\providecommand{\lemmaname}{Lemma}
\providecommand{\remarkname}{Remark}
\providecommand{\conjecturename}{Conjecture}
\providecommand{\definitionname}{Definition}
\providecommand{\corollaryname}{Corollary}
\allowdisplaybreaks

\def\bra#1{\langle{#1}\vert}
\def\ket#1{\vert{#1}\rangle}

\def\BraVert{e.g.,roup\,\mid\,\bgroup}

\newcommand{\ketbra}[2]{\ensuremath{|{#1}\rangle\langle{#2}|}}

\def\tr#1{\mbox{tr}\left[{#1}\right]}
\newcommand{\ptr}[2]{\mbox{tr}_{#1}\left[ #2 \right]}

\DeclareMathOperator{\diag}{diag}

\begin{document}


\title{Exponential improvement for quantum cooling through finite-memory effects}
\date{\today}

\author{Philip Taranto}
\email{philip.taranto@oeaw.ac.at}
\affiliation{Institute for Quantum Optics and Quantum Information, Austrian Academy of Sciences, Boltzmanngasse 3, 1090 Vienna, Austria}

\author{Faraj Bakhshinezhad}
\email{faraj.bakhshinezhad@oeaw.ac.at}
\affiliation{Institute for Quantum Optics and Quantum Information, Austrian Academy of Sciences, Boltzmanngasse 3, 1090 Vienna, Austria}

\author{Philipp Sch{\"u}ttelkopf}

\author{Fabien Clivaz}
\email{fabien.clivaz@oeaw.ac.at}
\affiliation{Institute for Quantum Optics and Quantum Information, Austrian Academy of Sciences, Boltzmanngasse 3, 1090 Vienna, Austria}

\author{Marcus Huber}
\email{marcus.huber@univie.ac.at}
\affiliation{Institute for Quantum Optics and Quantum Information, Austrian Academy of Sciences, Boltzmanngasse 3, 1090 Vienna, Austria}


\begin{abstract}
Practical implementations of quantum technologies require preparation of states with a high degree of purity---or, in thermodynamic terms, very low temperatures. Given finite resources, the Third Law of thermodynamics prohibits perfect cooling; nonetheless, attainable upper bounds for the asymptotic ground state population of a system repeatedly interacting with quantum thermal machines have been derived. These bounds apply within a memoryless (Markovian) setting, in which each refrigeration step proceeds independently of those previous. Here, we expand this framework to study the effects of memory on quantum cooling. By introducing a memory mechanism through a generalized collision model that permits a Markovian embedding, we derive achievable bounds that provide an exponential advantage over the memoryless case. For qubits, our bound coincides with that of heat-bath algorithmic cooling, which our framework generalizes to arbitrary dimensions. We lastly describe the adaptive step-wise optimal protocol that outperforms all standard procedures.
\end{abstract}
\maketitle

\section{Introduction}\label{sec:introduction}
Cooling a physical system is a thermodynamic task of fundamental and practical importance~\cite{Linden_2010,Goold_2016,Masanes_2017,Wilming_2017,Binder2018,Guryanova2020}. On the foundational side, the cooling potential is limited by the Third Law of thermodynamics, which posits the necessity of an infinite resource to be able to cool perfectly~\cite{Freitas2018}. This resource is subject to trade-offs: absolute zero is attainable in finite time given an infinitely-large environment; alternatively, given a finite energy source, one can only perfectly cool asymptotically. Practically, one cannot utilize an infinite resource, so the concern turns to: how cold can a system be prepared given resource constraints?

Formulating a theory with such constraints is typically scenario-dependent; nonetheless, one aims to develop theories that are widely applicable. For example, resource theories of quantum thermodynamics permit energy-conserving unitaries between the system and a thermal environment~\cite{Ng2018,Lostaglio_2019}. Analyzing the transformations for various environments and dynamical structures illuminates thermodynamic limitations. 

Recent work has examined the task of quantum cooling in such a setting~\cite{Clivaz_2019L,Clivaz_2019E}; the main result posits a universal bound for the ground state population of the system in the infinite-cycle limit. However, these results are derived in a memoryless (Markovian) setting, which is often not well-justified in experimental platforms where memory effects can affect the performance. For instance, Landauer's principle~\cite{Landauer_1961} can be violated in the non-Markovian regime~\cite{Pezzutto_2016,Man_2019}. 

A natural follow-up is to examine the role of memory in quantum cooling. Depending on the task and level of control, memory effects can have a detrimental or advantageous impact~\cite{Ishizaki_2009,Huelga_2013,Schmidt_2015,Bylicka_2016,Iles-Smith_2016,Kato_2016,Cerrillo_2016,Basilewitsch_2017,Naghiloo_2018,Fischer_2019}; nonetheless, applications highlight the potential to be unlocked by controlling the memory via reservoir engineering~\cite{Biercuk2009,Barreiro2010,Geerlings2013}. Attempts to generalize thermodynamics to the non-Markovian setting include trajectory-based dynamical unravelings~\cite{Strunz_1999,Jack_2000} and those based on the operational process tensor formalism~\cite{Pollock2018A,Romero_2018,Romero_2019,Strasberg_2019E1,Strasberg_2019E2,Strasberg_2019L,Romero_2020}, among others~\cite{Strasberg_2016,Whitney_2018}. However, such general approaches typically obscure insight regarding the crucial resources; it is often unclear whether reported ``quantum advantages'' are due to genuinely quantum effects (e.g., coherence) or memory. 

Here, we propose a mechanism for memory through a generalized collision model~\cite{Ciccarello2013,Lorenzo2017,Lorenzo2017A}, which---while not fully general---permits fair comparison between various memory structures. We show that in the asymptotic limit, the memory depth of the protocol plays a critical role and leads to exponential improvement over the Markovian case. Our results coincide with the limits of heat-bath algorithmic cooling protocols~\cite{Boykin2002,Schulman2005,Baugh2005,RaeisiPRL2015,Rodriguez-Briones2016,Alhambra_2019,RaeisiPRL2019,Kose_2019,Rodriguez-Briones2020Thesis} for qubit targets and our framework both unifies and generalizes this setting, applying to all system and environment structures. 


\section{Task: Cooling a quantum system}\label{sec:task}

A physical system is never isolated, which necessitates working within the theory of open systems, where the joint system -environment are closed, but environmental degrees of freedom are disregarded. Arbitrary environments permit perfect cooling with finite resources, as any physical transformation on a quantum system can be realized unitarily with a sufficiently-large environment; thus, further restrictions are necessary. 

We consider a system, $S$, and environment, $E$, with Hamiltonians $H_S$ and $H_E$, respectively. The system and environment begin uncorrelated and in equilibrium at inverse temperature $\beta := \tfrac{1}{k_B T}$. The joint system-environment evolves unitarily, with the system dynamics between the initial time and a later one $t$ described by the dynamical map, $\varrho^{(t)}_S(\beta) := \Lambda^{\tiny{(t)}}[\tau_S^{\tiny{(0)}}(\beta)]$, defined such that:
\begin{align}\label{eq:dynamicalmapmain}
    \varrho^{(t)}_S(\beta) = \ptr{E}{U^{(t)} (\tau_S^{\tiny{(0)}}(\beta) \otimes \tau_E^{\tiny{(0)}}(\beta))U^{\tiny{(t)} \dagger}},
\end{align}
where $\tau_X(\beta)$ denotes the thermal state of system $X$ at inverse temperature $\beta$, i.e., $\tau_X(\beta) := \mathcal{Z}_X^{-1}(\beta) \exp(-\beta H_X)$ with partition function $\mathcal{Z}_X(\beta) := \tr{\exp(-\beta H_X)}$.

The aim is to prepare $\varrho^{(t)}_S(\beta)$ as cold as possible. Cooling a system, however, can have several meanings: for one remaining in equilibrium, it could mean driving it to a thermal state of lower temperature; otherwise, one could consider increasing its ground state population or purity, or decreasing its entropy or energy. As such notions are generally nonequivalent, any study of cooling depends on the objective function~\cite{Clivaz_2019L}. We focus on achieving states that majorize all other potential states in the joint unitary orbit; this ensures optimization of all Schur-convex/concave functions of the vector of populations ordered with respect to non-decreasing energy eigenstates, in particular all above notions of temperature.

\section{Framework: Collision models with memory}\label{sec:framework}

Above we have described one step of a cooling protocol. In thermodynamic tasks, however, one is oftentimes interested in the multiple-cycle behavior. Here, one faces a choice in how to proceed: one could implement each operation independently of those previous, i.e., completely refresh the environment between steps, leading to Markovian dynamics; or, one could temporally correlate the cycles, leading to non-Markovian dynamics. The main difficulty in treating the latter is that memory effects can arise in various ways: they can be the manifestation of initial correlations, recurring system-environment or intra-environment interactions; or any combination thereof. In any case, for multiple cycles, the dynamical map in Eq.~\eqref{eq:dynamicalmapmain} fails to completely describe the system dynamics, since system-environment correlations can influence later evolution, in contradistinction to the Markovian setting, where the environment is entirely forgotten between steps. In general, one must track all system-environment degrees of freedom to describe the system evolution, which becomes unfeasible. Thus, we seek a framework that permits tractable memory and fair comparison between different memory structures. 

We propose a microscopic model for the environment and its interactions with the system. We consider a $d_S$-dimensional system with $H_S = \sum_{i=0}^{d_S-1} E_{i} \ket{i} \bra{i}_S$ and assume the environment comprises a number of identical units---which we call \emph{machines}---each being a $d_M$-dimensional quantum system with associated Hamiltonian $H_M = \sum_{i=0}^{d_M-1} \mathcal{E}_{i} \ket{i} \bra{i}_M$. We order Hamiltonians with respect to non-decreasing energies, and set $E_0=\mathcal{E}_0=0$ and $\mathcal{E}_\mathrm{max} =\mathcal{E}_{d_M-1}$. Assuming that the dynamics proceeds via successive unitary ``collisions'' between the system and subsets of machines yields a collision model with memory. 

The memory effects that arise from endowing such models with various dynamical structures have been examined: considerations include initially correlated machines~\cite{Rybar2012, Bernardes2014}, inter-machine~\cite{Ciccarello2013, McCloskey2014, Lorenzo2016, Cakmak2017, Campbell2018} or repeated system-machine collisions~\cite{Grimsmo2015, Whalen2017}, or hybrid variations~\cite{Kretschmer2016,Lorenzo2017A, Lorenzo2017,Ciccarello2017}. In certain cases, the model exhibits finite-length memory~\cite{Taranto_2019L,Taranto_2019A,Taranto2019S,TarantoThesis}. In the limit of many machines, the system is expected to interact with only mutually-exclusive subsets of machines; since any used machines never play a subsequent role, one yields a microscopic picture of Markovian evolution that gives a Lindbladian master equation in the continuous-time limit~\cite{Rau1963,Ziman2002,Scarani2002,Ziman2005}. 

\begin{figure}[t!]
\centering
\includegraphics[width=\linewidth]{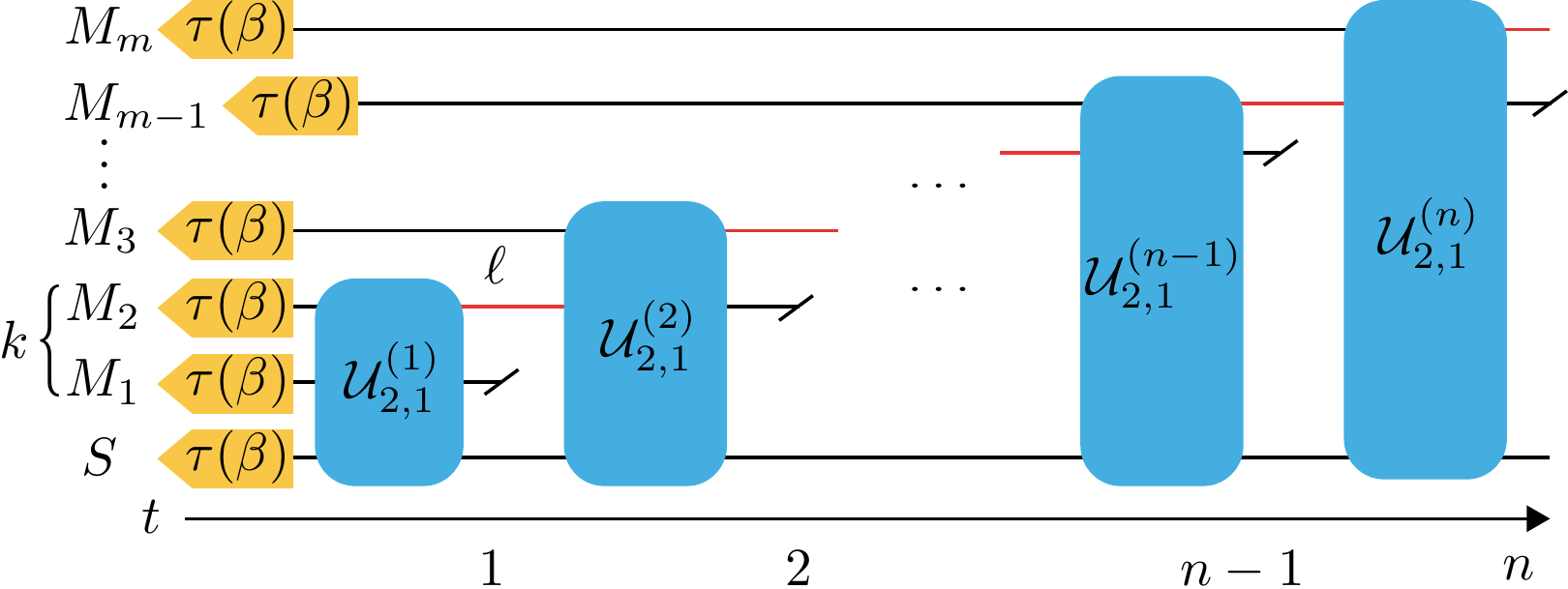}
\caption{ \label{fig:collisioncircuit}\textit{Collision model with memory.} At each step, the system $S$ interacts unitarily with $k$ machines, of which $\ell$ carry forward memory (red lines). Here we illustrate $k=2, \ell=1$, with $m$ the total number of machines used by timestep $n$.}
\end{figure} 

Although not fully general, this setting captures tractable non-Markovian dynamics. In this Letter, we will analyze the memory effects that arise from repeated system-machine interactions (see Fig.~\ref{fig:collisioncircuit}). More precisely, we consider $k$ machines to interact with the system between timesteps, with some $\ell \leq k$ of these carrying memory forward; this reduces to a Markovian protocol involving $k$ machines for $\ell=0$. The assumptions are that the system and all machines begin uncorrelated, and there are no interactions between memory-carrying machines and fresh ones other than those involving the system. These are valid whenever the memory-carrying machines relax much slower than those that rethermalize between steps. We can vary the number of machines in each interaction, $k$~\footnote{One could consider the restricted case of $(k+1)$-partite system-machine interactions that are decomposable into sequences of $p$-partite interactions for $p \leq k$. We do not make this restriction and allow any multi-partite interaction to be genuinely so.}, the number of memory carriers, $\ell$, the initial temperature, $\beta$, and the Hamiltonians.

This generic framework applies to a wide range of protocols. For instance, one can compare adaptive strategies, where different unitaries are performed between steps, versus non-adaptive ones, where a fixed dynamics is repeated. Additionally, one can restrict the allowed unitaries, such as limiting the set from general ``coherent'' ones (that require an external energy source) to ``incoherent'' energy-conserving transformations (where the cooling resource is an additional hot bath)~\cite{Clivaz_2019E,Aberg_2014}. Lastly, one could allow the memory structure itself to be adaptive, where $k$ and $\ell$ vary between times; we do not consider this and instead focus on cooling limits for fixed structures. A choice of $k$ and $\ell$, along with the system and machine dimensions, determines the control complexity afforded to the experimenter: intuitively, $k$ is related to spatial complexity and $\ell$ to temporal. We now compare the achievable cooling of a system for different memory structures. 


\section{Memory-enhanced cooling}\label{sec:main}

The fundamental Markovian cooling bounds have been derived in Refs.~\cite{Clivaz_2019L,Clivaz_2019E}. The optimally-cool system state at any finite time depends upon the energy-level structure between the system and machines and the level of control. However, in the asymptotic limit of Markovian operation, the vector of eigenvalues of the asymptotic state (in any aforementioned control paradigm) is majorized by that of
\begin{gather}\label{eq:markovianasymptoticstate}
    \varrho^*_S(\mathcal{E}_\textup{max},\beta,k) = \sum_{n=0}^{d_{S} -1}\frac{e^{-\beta n k{\mathcal{E}}_{\textup{max}}}}{\mathbbm{Z}_{\tiny{S}}(\beta,k\mathcal{E}_\textup{max})} \, \ket{n}\bra{n}_S,
\end{gather}
whenever the initial state $\tau_S^{(0)}(\beta)$ is majorized by $\varrho^*_S(\mathcal{E}_\textup{max},\beta,k)$; here $\mathbbm{Z}_{\tiny{X}}(\beta,\mathcal{E}):=\sum_{n=0}^{d_{\tiny{X}} -1}e^{-\beta n \mathcal{E}}$ is a quasi-partition function (depending only on the maximum energy gap of each machine). The state in Eq.~\eqref{eq:markovianasymptoticstate} is attainable with coherent control, positing the ultimate Markovian cooling limit.

The intuition is that the optimal protocol reorders the eigenvalues of the system and relevant machines at each step such that the maximum population is placed into the ground state subspace of the system, the second largest into the first excited state subspace, and so on. When this cycle is repeated with fresh machines at each timestep, the asymptotic state looks as if it had interacted with only the qubit subspace of each machine with maximum energy difference. However, the result cannot immediately be extended to the non-Markovian regime, as its derivation relies on an inductive argument on the system state at each step; for non-Markovian dynamics, this cannot be expressed in terms of the previous state, posing a logical roadblock. 

Whenever $\ell>0$ the generalized collision model is non-Markovian. Nonetheless, a relevant result states that such non-Markovian collision models can be lifted to a Markovian dynamics on a larger state space~\cite{Campbell2018}. For a system interacting with $k$ machines at each step, of which $\ell$ feed forward, the dynamics can be embedded into a Markovian one by considering the system and $\ell$ memory carriers as a unified system, which interacts at each step with $k-\ell$ fresh machines; such a process is said to have \emph{memory depth} $\ell$. In Appendix~\ref{app:markovianembedding}, we detail the Markovian embedding, which leads to the following results. 
 

\subsection{Asymptotic cooling advantage}\label{subsec:asymptotic}

We now present the universal cooling bound for the non-Markovian collision model in the infinite-cycle limit:

\begin{thm}\label{thm:asymptoticmain}
For any $d_S$-dimensional system interacting at each step with $k$ identical $d_M-$dimensional machines, with $\ell$ of the machines (labeled $L$) used at each step carrying the memory forward, in the limit of infinitely many cycles:

\noindent i) The ground state population of $S$ is upper bounded by\vspace{-0.4em}
    \begin{gather}
    \label{eq: population bound kl}
    {p}^*(\mathcal{E}_\textup{max},\beta,k,\ell)= \big(\sum_{n=0}^{d_{{S}} -1}e^{-\beta n d_M^\ell(k-\ell){\mathcal{E}}_{\textup{max}}}\big)^{-1}.
\end{gather}
\noindent ii) The vector of eigenvalues of the output system state is majorized by that of the following attainable state 
\begin{gather}
 \label{eq: asymptotic state kl}
    \varrho^*_S(\mathcal{E}_\textup{max},\beta,k,\ell)\!=\!\hspace{-0.2em}\sum\limits_{n=0}^{d_{S} -1}\!\frac{e^{-\beta n d_M^\ell(k-\ell){\mathcal{E}}_{\textup{max}}}}{\mathbbm{Z}_{\tiny{S}}(\beta,d_M^\ell(k-\ell){\mathcal{E}}_{\textup{max}})}\,\ket{n}\bra{n}_{\tiny{S}},\!
\end{gather}
whenever the initial state $\tau_S(\beta)\otimes\tau_{M}(\beta)^{\otimes \ell}$ is majorized by 
\begin{gather}\label{eq:slmajorizedinitstate}
 \varrho^*_{SL}(\tiny{\mathcal{E}_\textup{max},\beta,k,\ell})\!=\!\hspace{-1em}\sum_{n=0}^{\tiny{d_{SL}-1}}\tiny{ \hspace{-1em}\frac{e^{-\beta n (k-\ell){\mathcal{E}}_{\textup{max}}}}{\tiny{\mathbbm{Z}_{\tiny{SL}}(\beta, (k-\ell)\mathcal{E}_\textup{max})}}}\, \ket{n}\bra{n}_{\tiny{SL}}.\!
\end{gather}
\end{thm}

\noindent \emph{Sketch of Proof.} We use the Markovian embedding to lift the non-Markovian dynamics of the target $S$ to a Markovian process for the target-plus-memory carriers $SL$ system, which interacts with $k-\ell$ fresh machines (which we label $R$) at each step. This implies that optimally cooling $SL$ is necessary to optimally cool $S$. From Ref.~\cite{Clivaz_2019L}, the asymptotically-optimal state of $SL$ has the same eigenvalue distribution as Eq.~\eqref{eq:slmajorizedinitstate}, whenever the initial $SL$ state is majorized by $\varrho^*_{SL}$, and is thus unitarily equivalent to it. As majorization concerns all partial sums, given that initial condition, whatever protocol one chooses to cool $S$, the asymptotic state cannot be colder than $\varrho_S^*$ (which is the coldest $S$ state in the unitary orbit of $\varrho_{SL}^*$). This implies that the asymptotic ground state population is upper bounded by $p^*$. See Appendix~\ref{app:proof}. \qed

There are many noteworthy points: firstly, the optimal ground state population is enhanced by $d_M^\ell$ compared to the Markovian case, highlighting the drastic role of memory; in particular, one achieves an exponential improvement in $\ell$. Secondly, as the factors in Eq.~\eqref{eq: asymptotic state kl} arise independently from various sources (i.e., $S$, $L$ and $R$), the bound extends to the case where $L$ is an arbitrary $d_L$-dimensional system and $R$ an arbitrary $d_M$-dimensional system (with maximum energy gap $\mathcal{E}_{\text{max}}'$), with $d_M^{\ell} \rightarrow d_L$ and $(k-\ell) \mathcal{E}_{\text{max}} \rightarrow \mathcal{E}_{\text{max}}'$. This clarifies that the asymptotic bound only depends on the dimension of $L$, not on its energy structure. Lastly, the asymptotic $SL$ state of Eq.~\eqref{eq:slmajorizedinitstate} is unitarily equivalent to a tensor product state that has Eq.~\eqref{eq: asymptotic state kl} as its reduced state on $S$. Nonetheless, throughout the cooling protocol correlations build up, due to the finite-time dependence on the energy structures of the systems involved, before dying out asymptotically; in Appendix~\ref{app:correlations}, we explore the role of correlations in more detail.

Returning to Theorem~\ref{thm:asymptoticmain}, Eq.~\eqref{eq: asymptotic state kl} allows us to compare limits for various $k,\ell, \beta$ and $\mathcal{E}_\textup{max}$ (see Fig.~\ref{fig:mainplot}):
\begin{cor}\label{cor:hierarchy}
The asymptotic hierarchy is determined via:
\begin{align}
    {\varrho}_S^*(&\mathcal{E}_\textup{max},\beta,k,\ell) {\prec}({\succ}){\varrho}_S^*(\mathcal{E}_\textup{max}',\beta',k',\ell') \notag  \\
  \textup{if}&~~ \beta(k-\ell)d_M^\ell \mathcal{E}_\textup{max}\, \leq (>)\, \beta'(k'-\ell')d_M^{\ell'}\mathcal{E}_\textup{max}'. 
\end{align}
\end{cor}

\begin{figure}[t!]
\centering
\includegraphics[width=\linewidth]{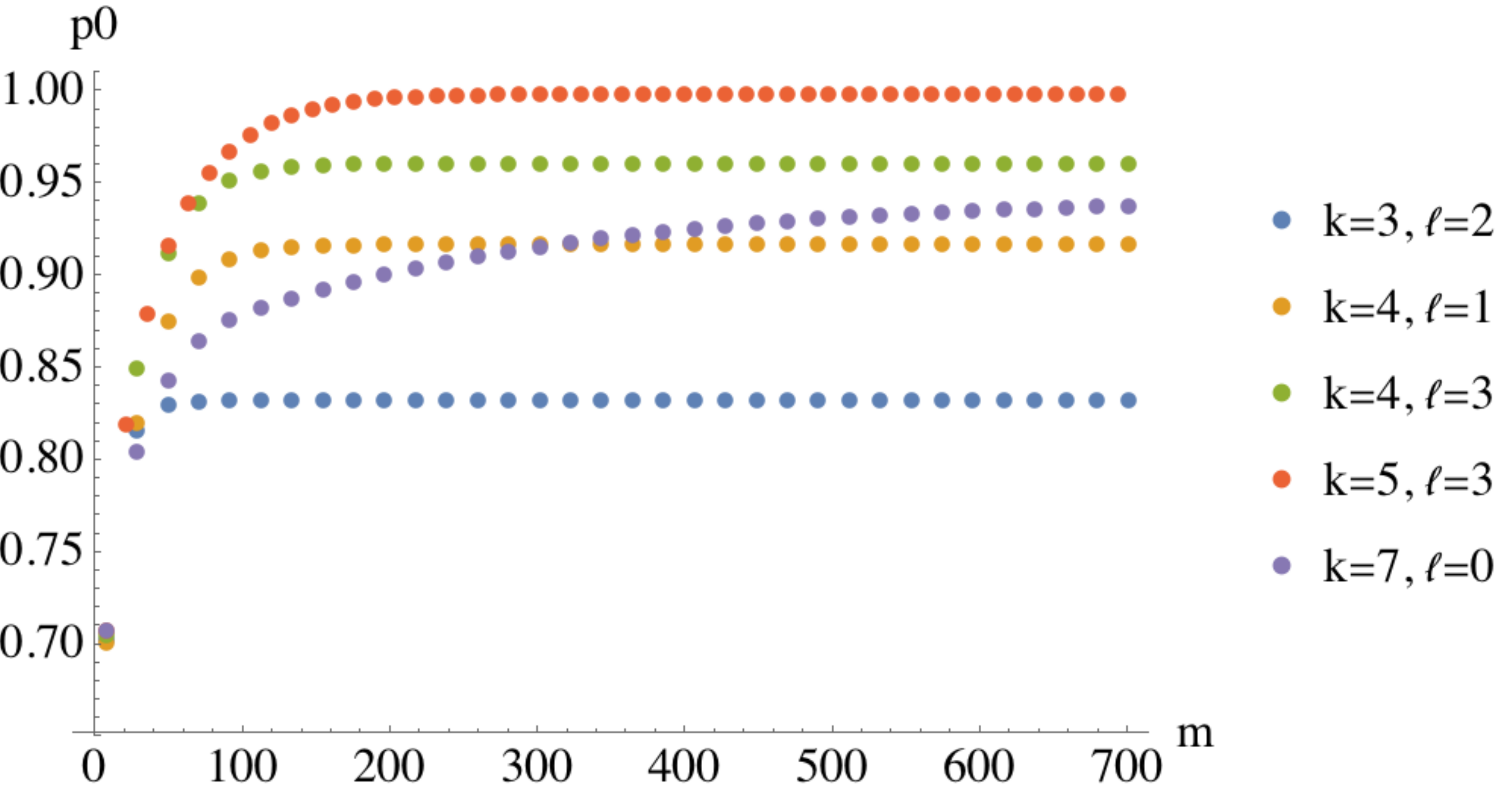}
\caption{ \label{fig:mainplot}\textit{Cooling behavior.} We simulate the ground state population after $m$ machines have been used in the step-wise optimal cooling protocol (see Theorem~\ref{thm:stepwiseoptimalmain}) for a qubit system and machines with fixed $\beta = 0.2, E_{\textup{max}} = 1$ and $\mathcal{E}_{\textup{max}} = 2$. The asymptotic hierarchy agrees with Corollary~\ref{cor:hierarchy}, whereas the complex short-term behavior can exhibit crossovers.} 
\end{figure} 

\subsection{Step-wise optimal protocol}\label{subsec:stepwiseoptimal}

The bound is achievable and one protocol to do so reorders the global eigenspectrum at each step such that they are non-increasing with respect to non-decreasing energy eigenstates of $SL$. In the last step, the protocol additionally reorders the eigenvalues of the obtained $SL$ state largest-to-smallest with respect to non-decreasing energy eigenstates of $S$. Precisely, at each step $j$ the system and memory carriers are optimally cooled via a unitary $V_{SLR}^{(j)}: \varrho_{SLR}^{(j)}=V_{SLR}^{(j)} \varrho_{SLR}^{(j-1)} V_{SLR}^{(j) \dagger}$ that acts as~\footnote{Here, $R$ refers to the fresh machines included at each step, with an implied identity map on all other systems.}
\begin{align}\label{eq:sortsl}
    \varrho_{SLR}^{(j)} = \sum_{\xi=0}^{d_S d_M^\ell - 1} \sum_{\omega=0}^{d_M^{k-\ell} - 1} \lambda^{\downarrow}_{\xi \cdot d_M^{k-\ell} + \omega} \ket{\xi \omega} \bra{\xi \omega},
\end{align}
where $\ket{\xi \omega} = \ket{\xi}_{SL}\otimes\ket{\omega}_R$ and $\lambda^{\downarrow}$ denotes the eigenvalues of $\varrho_{SLR}^{(j-1)}$ in non-increasing order. This unitary dissipates maximal heat into the machines that play no subsequent role, and thus at any finite timestep $j$, the protocol has achieved the coldest $SL$ state possible given its history, which is crucial for finite-time optimality. By implementing the sequence $\{ V_{SLR}^{(j)} \}_{j=1,\hdots,n}$, although the final (timestep $n$) $SL$ state $\varrho_{SL}^{(n)}$ generically exhibits correlations, there always exists a unitary $W^{(n)}_{SL}$ that ensures $S$ is optimally cool by further reordering the eigenvalues under the previous constraint; in the asymptotic limit, i.e., when $n \to \infty$, said unitary completely decorrolates $SL$ (whereas at finite times, some correlation generically remains). However, while this strategy attains the optimally-cool $S$ for any final timestep $n$, this protocol is not necessarily \emph{step-wise} optimal.

To derive the step-wise optimal protocol, consider the unitary $W_{SL}^{(j)}: \sigma_{SL}^{(j)}=W_{SL}^{(j)} \varrho_{SL}^{(j)} W_{SL}^{(j) \dagger}$ that acts as
\begin{align}\label{eq:sorts}
    \sigma_{SL}^{(j)} = \sum_{\mu=0}^{d_S - 1} \sum_{\nu=0}^{d_M^{\ell} - 1} \lambda^{\downarrow}_{\mu \cdot d_M^{\ell} + \nu} \ket{\mu \nu} \bra{\mu \nu},
\end{align}
where $\ket{\mu \nu} = \ket{\mu}_{S}\otimes\ket{\nu}_L$ and $\lambda^{\downarrow}$ here denotes the eigenvalues of $\varrho_{SL}^{(j)}$ in non-increasing order. $W_{SL}^{(j)}$ optimally cools $S$ given any $SL$ state by unitarily transferring maximal entropy towards $L$; thus, if we apply $W_{SL}^{(j)}$ \emph{at each step $j$} after having optimally cooled $SL$ via $V^{(j)}_{SLR}$ until then, i.e., implement $U^{(j)}_{SLR} = W^{(j)}_{SL} V^{(j)}_{SLR}$, $S$ is guaranteed to be optimally cool. This leads to the following, proven in the Appendix~\ref{app:finitetime}, where we examine finite-time behavior.

\begin{thm}[Step-wise optimal cooling protocol]\label{thm:stepwiseoptimalmain}
By applying $U^{(j)}_{SLR}$ described above at each step, the cooling protocol is step-wise optimal regarding the temperature of the system.
\end{thm}

\section{Relation to heat-bath algorithmic cooling}\label{sec:hbac}

Above we have derived the cooling limit in a controlled non-Markovian setting; through the Markovian embedding, we can further make direct connection with heat-bath algorithmic cooling \textbf{(HBAC)}~\cite{Boykin2002,Schulman2005,Baugh2005,RaeisiPRL2015,Rodriguez-Briones2016,Alhambra_2019,RaeisiPRL2019,Kose_2019,Rodriguez-Briones2020Thesis}, the limitations of which align with our results for qubit targets. Here, one cools a ``target'' system by cooling a larger ensemble of ``compression'' systems~\footnote{Together, the target and compression systems constitute what is often called the ``computation'' system.} via interactions with ``reset'' systems that rethermalize between steps. This permits better cooling than cooling the target alone with only reset systems as a resource; indeed, HBAC protocols are non-Markovian and a special case of our framework, which treats the compression/refrigerant systems as memory-carrying machines and the reset systems as fresh machines, as detailed below and in Appendix~\ref{app:hbac}.

i) Each of the target, memory carrier (compression/refrigerant), and reset systems can comprise multiple subsystems of arbitrary dimension, with arbitrary energy spectra and initial temperatures, which determines the asymptotic hierarchy for different strategies. In contrast, many HBAC studies focus only on target and reset qubits~\cite{Boykin2002,Schulman2005,RaeisiPRL2019}; although some consider qudit compression~\cite{Rodriguez-Briones2016} and reset systems~\cite{RaeisiPRL2015}, no HBAC study has shown results pertaining to the general qudit-qudit-qudit case. ii) Our results are based on majorization (as are those in Refs.~\cite{Schulman2005,RaeisiPRL2015}), and therefore applicable to more general notions of cooling than the often-considered ground state population (e.g., in Refs.~\cite{Boykin2002,Baugh2005,Rodriguez-Briones2016}), which crucially differ for high-dimensional systems~\cite{Clivaz_2019L}. This is important for quantum computing---for which cooling is a critical requirement---where high-dimensionality can simplify logical structures~\cite{Lanyon2009,Babazadeh2017,Imany2019}. iii) Finally, our results extend the partner-pairing algorithm---introduced to maximize the ground state population of a qubit in Ref.~\cite{Schulman2005}---to the most general setting (described above). The partner-pairing algorithm is step-wise optimal with a complexity that scales polynomially; our protocol achieves the same scaling, as the sorting required at each step can be achieved with a single operation. Although these operations depend on the global state at each step, in Appendix~\ref{app:hbac} we present a simple robust algorithm (based on one presented in Ref.~\cite{RaeisiPRL2019}) that uses only a fixed state-independent two-body interaction to reach the asymptotically-optimal state.

By contextualizing HBAC within the framework of collision models with memory, our work provides both a unification and generalization of HBAC. Moreover, our approach lends itself to modeling realistic HBAC experiments, where reset systems only partially thermalize, as considered in Ref.~\cite{Alhambra_2019}.


\section{Conclusions}\label{sec:conclusions}

In this Letter, we have put forward a framework for consistently dealing with memory when cooling quantum systems; indeed, the generalized collision model proposed is versatile enough to analyze the role of memory in various thermodynamic tasks. In doing so, we have revealed the potential for exponential improvement in the reachable ground state population (and more general notions of cooling), yielding drastic enhancement already for modest memory depths. Through a Markovian embedding of our framework, we could connect our framework with HBAC, demonstrating the latter to be a particular class of non-Markovian dynamics. Our results can be read as a generalization of HBAC applicable to arbitrary target and compression systems and bath spectra; by putting all HBAC protocols on an equal footing, our work opens the door to comparative studies that can now be made fairly. Moreover, we clarify the origin of the advantages that make HBAC so effective. Together with that of Refs.~\cite{Clivaz_2019L,Clivaz_2019E}, our work unifies HBAC with the resource theory of thermodynamics, as all results can be achieved either via coherent control or energy-conserving unitaries on enlarged systems. 

The exponential improvement with respect to the memory carriers stands in contrast to the only linear enhancement in the number of rethermalizing systems, highlighting the importance of controllable memory. Thus, given the ability to perform $(k +1)$-partite interactions, having $\ell = k-1$ of these systems as memory carriers is both the minimum requirement (as the system-and-memory must be open, otherwise any cooling ability rapidly diminishes) and, moreover, the optimal configuration. In particular, this implies that one can achieve the exponential advantage via interactions involving the system and memory carriers and \emph{only one additional} reset system. Of course, if $\ell$ is large, the ability to implement complex many-body interactions presents a difficult challenge. To this end, we have developed an explicit protocol (see Appendix~\ref{app:hbac}) which necessitates only a fixed, two-body interaction to achieve the fundamental bound. In particular, this implies that the attainability of the optimal asymptotic ground state population does not require implementing highly non-local unitaries. In fact, this robust protocol applies to arbitrary dimensional systems for the target, memory and reset parts, which is significant because high-dimensional systems are becoming increasingly relevant for fault-tolerant quantum computing~\cite{Lanyon2009,Babazadeh2017,Imany2019}---a major motivation for cooling quantum systems in the first place.

Our results consolidate the limits for quantum refrigeration in a setting with perfect control and high-quality isolation of the target and memory carriers. However, in most experimental scenarios, further challenges arise. We have assumed that the uncontrolled system-environment interactions are negligible compared to the controlled ones. For finite times, our results are reliable due to the exponential scaling, which makes it sufficient to run the protocol for a short time to approximate the asymptotics. An immediate concern is the impact of uncontrolled interactions: either to model imperfect target isolation or to better understand the realistic asymptotics, as for infinite steps, rethermalization of the target cannot always safely be neglected. Another assumption worth analyzing is that of perfect control: to implement a unitary perfectly, one requires both precise clocks~\cite{Malabarba_2015}, which have their own thermodynamic costs~\cite{Erker_2017,woods2019resource}, and high control over the interaction terms. While this is plausible for quantum computing devices, other systems are more challenging to control, particularly those with multi-partite interactions that are only perturbatively accessible~\cite{Mitchison_2016}. Lastly, a resource-theoretic approach has derived the minimum amount of energy required to implement transformations such as those considered here, which should also be accounted for~\cite{Chiribella2019}. Our framework lends itself to such pragmatic analyses of cooling; deriving similar bounds in realistic settings highlights the potential to elevate our results beyond fundamental limits and towards practical guidelines for quantum experiments. 

\begin{acknowledgments}
We thank Ralph Silva for discussions. This work was supported by the Austrian Science Fund (FWF) projects: Y879-N27 (START) and P31339-N27 and the FQXi Grant: FQXi-IAF19-03-S2.
\end{acknowledgments}

\onecolumngrid
\appendix

\section{Markovian embedding of collision models with memory}\label{app:markovianembedding}

We are interested in exploring analytically the effects of memory regarding the task of cooling a quantum system. We do not wish to allow for arbitrary non-Markovianity, as this would lead to an infinite resource in a sense that it allows us to cool the system to the ground state perfectly. Rather aim to obtain a cooling bound in the limit of infinite cycles for a particular class of non-Markovian dynamics, namely a generalized collision model endowed with memory. Such collision models with memory are quite general, simply assuming that between each step of the dynamics, the system interacts with $k$ constituent sub-machines (which altogether make up the full set of machines), of which some $\ell<k$ of these carry the memory forward. The Markovian setting is recovered for $\ell=0$. A schematic is provided in Fig.~\ref{fig:collisionschema}.

\begin{figure}[b!]
\centering
\includegraphics[width=\linewidth]{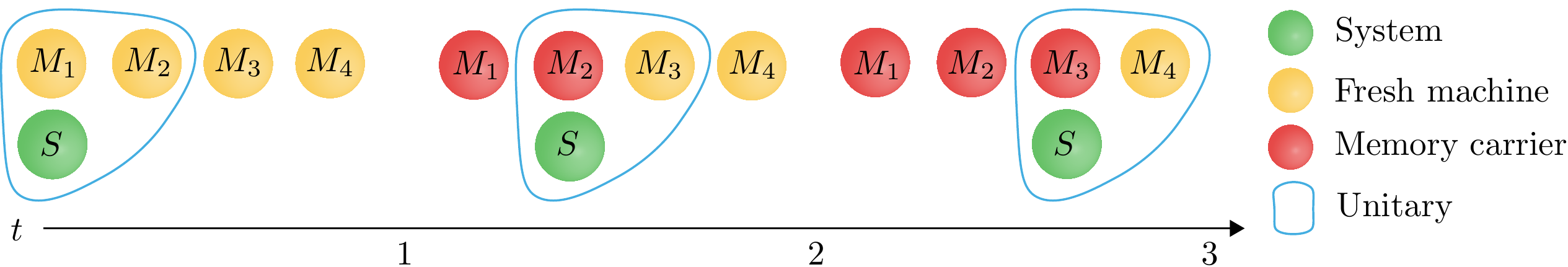}
\caption{ \label{fig:collisionschema}\textit{Generalized collision model with memory.} In the collision model picture, the environment that the system (green) interacts with is assumed to be comprised of individual sub-units, which we call machines. Between each step of dynamics, the system interacts unitarily (blue outline) with a subset of these; the cardinality of this set is labeled by $k$ throughout this work. A further subset of these previously-interacted-with machines of cardinality $\ell<k$ take part in the following interaction, becoming memory carriers (red). At each timestep, $k-\ell$ fresh machines are incorporated into the dynamics (yellow). Here we have shown $k=2, \ell=1$. }
\end{figure} 

We consider a target system of dimension $d_S$ and local Hamiltonian $H_S=\sum_{i=0}^{d_S -1}E_i \ketbra{i}{i}_S$, where $\{E_i\}$ are sorted in non-decreasing order and an environment comprising of a number of constituent identical machines of finite size $d_M$, each of which has the local Hamiltonian  $H_{M_j}=\sum_{i=0}^{d_{M_j} -1}\mathcal{E}_i\, \ketbra{i}{i}_{M_j}$, where $\{\mathcal{E}_i\}$ are also sorted in non-decreasing order. The system and all machines (i.e., the entire environment) begin uncorrelated in a thermal state with the same inverse temperature $\beta$:
\begin{align}
  \varrho_S^{(0)}(\beta) &=\tau_S(\beta) \quad \quad \text{and} \quad \quad
   \varrho_{M}^{(0)}(\beta)= \bigotimes_{j} \tau_{M_j}(\beta)
\end{align}
where $\tau_X(\beta):=\mathcal{Z}_X^{-1}\exp{(-\beta H_X)}$ with the partition function $\mathcal{Z}_X:=\mathrm{tr}[\exp{(-\beta H_X)}]$. 

Fixing $k$ and $\ell$ provides a particular dynamical structure of the non-Markovian process: it stipulates that at each timestep there are $k$ machines interacting with the system, of which $\ell$ are kept to perpetuate the memory. For example, after $n$ steps, the system state is
\begin{align}\label{eq:appcm}
 \varrho_S^{(n)}(\beta, k, \ell) = \mathrm{tr}_{M}[U_{k,\ell}^{(n)}\hdots U_{k,\ell}^{(1)}\, ( \tau_S(\beta)\bigotimes_{j=1}^{m} \tau_{M_j}(\beta) ) \,U_{k,\ell}^{(1) \dagger}\,\hdots \,U_{k,\ell}^{(n) \dagger}], 
\end{align}
where $U_{k,\ell}^{(n)}$ is an arbitrary unitary transformation between the target system and the $k$ machines labeled by $\{(n-1)(k-\ell)+1,\, \hdots , n(k-\ell)+\ell \}$ (an identity map is implied on the other machines) and $m:=k+(n-1)(k-\ell)$ is the total number of machines used by the protocol up until timestep $n$, which will be important in making finite time comparisons, as we do in Appendix~\ref{app:finitetime} (see Fig.~\ref{fig:collisioncircuit} for a graphical depiction in terms of a circuit diagram).

Importantly, the state of the system at any time is a function of the full microscopic energy structure $\{ E_i \}$ and $\{ \mathcal{E}_i\}$ (which we do not explicitly label for ease of notation), $\beta, k$ and $\ell$; the latter two numbers specify a particular dynamical structure in terms of which systems the unitaries act upon between timesteps. If $\ell=0$, the dynamics of the system is Markovian, since at each step, the system interacts with fresh machines that contain no memory of the past dynamics of the system. Otherwise, each of the machines interacts more than once with the target and only $k-\ell$ fresh machines are added into the interaction at each step. 

Eq.~\eqref{eq:appcm} highlights the restriction imposed by the assumption of generalized collision model dynamics from the fully general case of non-Markovian dynamics where the full system-environment must be tracked; in particular, a subset of the environment (the $k-\ell$ rethermalizing systems) is traced out between steps, rendering the dynamics tractable for small $k, \ell$. However, it is important to note that on the level of the system, memory effects still play a role. We first show that for $\ell > 0$ the dynamics considered is indeed non-Markovian in general.

To analyze the proposed setting, we need to look at the evolution of the entire joint system and machines to consider the effect of the memory in the protocol. For instance, consider rewriting Eq.~\eqref{eq:appcm} as a dynamical map taking the initial system state $\varrho_S^{(0)}(\beta) = \tau_S(\beta)$ to the later one under a generic dynamical structure determined by the choice of $k$ and $\ell$, i.e., define 
\begin{align}
    \Lambda^{(n:0)}_{k,\ell}(\beta)[ X_S] := \mathrm{tr}_{M}[U_{k,\ell}^{(n)}\hdots U_{k,\ell}^{(1)}\, ( X_S\bigotimes_{j} \tau_{M_j}(\beta) ) \,U_{k,\ell}^{(1) \dagger}\,\hdots \,U_{k,\ell}^{(n) \dagger}],
\end{align} 
where we have now included all of the machines in the environment and an identity map on those not taking part in the interactions until timestep $n$ is implied, such that
\begin{align}
    \varrho_S^{(n)}(\beta, k, \ell) = \Lambda^{(n:0)}_{k,\ell}(\beta)[ \varrho_S^{(0)}(\beta)].
\end{align}

Linearity, complete positivity and trace-preservation of the map $\Lambda^{(n:0)}_{k,\ell}(\beta)$ is guaranteed for any $\beta, k, \ell$ and $n$ by the fact that $S$ and $E$ begin initially uncorrelated and the dynamics evolves unitarily on the global level, before a final partial trace is taken over the environment degrees of freedom. Complete positivity is particularly important to ensure that the map takes valid quantum states to valid quantum states. In general, the global state $\varrho_{SM}^{(n)}(\beta, k, \ell)$, where $M$ labels the subset of the environment that has taken part non-trivially in the dynamics up until timestep $n$, involves correlations between $S$ and $M$; taking the partial trace over $M$ destroys all such correlations. Thus, one cannot, in general, describe the evolution of the system between multiple times as a divisible concatenation of completely positive and trace-preserving \textbf{(CPTP)} maps, i.e.,
\begin{align}\label{eq:appopcpdiv}
    \varrho_S^{(n)}(\beta,k,\ell) = \Lambda_{k,\ell}^{(n:0)}(\beta) [\varrho_S^{(0)}(\beta)] \neq \Lambda_{k,\ell}^{(n:t)}(\beta) \circ \Lambda_{k,\ell}^{(t:0)}(\beta) [\varrho_S^{(0)}(\beta)].
\end{align}

Here, we have defined $\Lambda_{k,\ell}^{(n:t)}(\beta)$ as the map that would be tomographically constructed if one were to discard the system at time $t$ (which is generally correlated to $M$) and perform a quantum channel tomography by preparing a fresh basis of input states (see Fig.~\ref{fig:cpdivisibility}); since the reprepared state is uncorrelated to $M$ by construction, the map $\Lambda_{k,\ell}^{(n:t)}(\beta)$ is guaranteed to be CPTP for any choice of parameters~\cite{milz_introduction_2017}. Testing for equality in Eq.~\eqref{eq:appopcpdiv} then corresponds to the operational notion of CP-divisibility proposed in Ref.~\cite{Milz2019CP}; importantly, its breakdown acts as a valid witness for non-Markovianity that is stricter than other notions of CP-divisibility proposed throughout the literature (in particular, it is stronger than that based on invertible CP-divisibility in any case where the dynamical maps are invertible). Of course, the fact that Eq.~\eqref{eq:appopcpdiv} is generally an inequality for generic dynamics does not imply that it is so for the particular optimal cooling dynamics described throughout this article; however, it is simple to show that the optimal cooling protocol indeed generates correlations between the system and machines that lead to a breakdown of (operational) CP-divisibility, and hence the particular dynamics considered throughout is inherently non-Markovian.

\begin{figure}[t!]
\centering
\includegraphics[width=\linewidth]{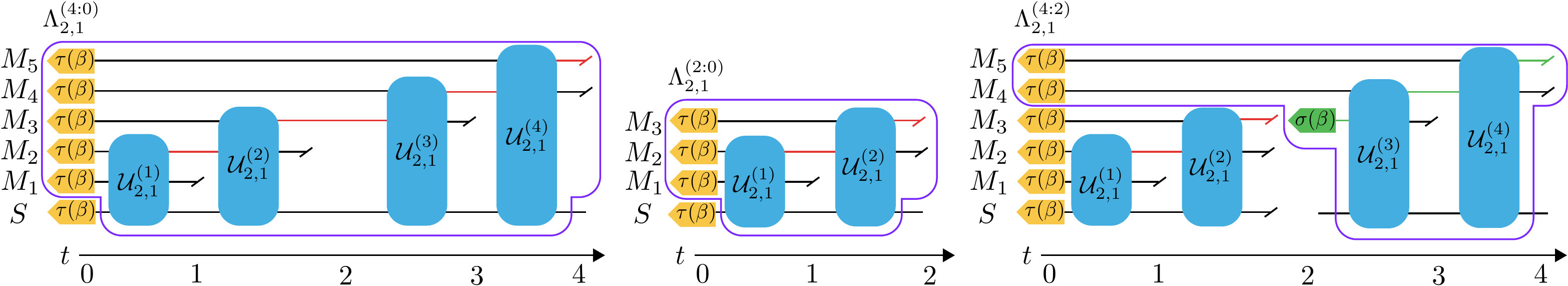}
\caption{ \label{fig:cpdivisibility}\textit{Operational CP-divisibility.} A breakdown of operational CP-divisibility is a witness for non-Markovianity. The test consists of tomographically constructing a set of maps describing the dynamics and checking the validity of Eq.~\eqref{eq:appopcpdiv}. We illustrate the scenario for a subset of times $t = 2, n =4$, with $k=2,\ell=1$: the left panel depicts the map $\Lambda^{(4:0)}_{2,1}$, which comprises everything within the purple border, i.e., the initial states of all machines, all joint unitary interactions, and the final partial trace. Note that the final memory carrier $M_5$ should continue forward, but the map $\Lambda^{(4:0)}_{2,1}$ does not capture this and traces out that subsystem. The middle panel depicts the map $\Lambda^{(2:0)}_{2,1}$. Both maps can be tomographically reconstructed by preparing a basis of input states at the initial time and measuring the outputs at time $t=4$ and $t=2$, respectively. As the system begins initially uncorrelated with the machines, the unitary dilation guarantees that the maps constructed are CPTP. The final map needed, $\Lambda^{(4:2)}_{2,1}$, is shown in the right panel. In general, at time $t=2$, the system is correlated to the machines, thus breaking the assumption of no initial correlations. An operational circumvention is to discard the system state at $t=2$ and reprepare a fresh one, thereby erasing all system-machine correlations. This has the effect of rendering the memory carriers into a fixed quantum state, which can be included in the description of $\Lambda^{(4:2)}_{2,1}$ to ensure that it is CP. When memory is present, tracing out the system at the intermediary timestep generally conditions the state of the memory carriers into a state that generally differs from the initial thermal $\tau(\beta)$, labeled here $\sigma(\beta)$, with the altered part of the evolution depicted in green; thus, the full dynamics generically differs from the concatenation. }
\end{figure} 

Nonetheless, the collision model memory structure that we have introduced crucially allows for a Markovian embedding that permits a significant simplification in the analysis~\cite{Campbell2018}. As mentioned previously, in general, one would need to track the total joint evolution throughout the entire protocol, which quickly becomes computationally exhaustive as $k$ grows. However, for a choice of $k$ and $\ell$, we can group the system $S$ and $\ell$ of the machines into a larger joint system, which we label $SL$, which interacts with $k-\ell$ fresh machines at each timestep; we label these fresh machines with $R$ as they model rethermalization of some of the machines with the environment. On the level of $SL$, the dynamics is Markovian, as the degrees of freedom carrying the memory have been included in the description of the target system. One can obtain the state of the overall $SL$ target by tracing out the $R$ machines at each step. We therefore have
\begin{equation}
 \varrho_{SL}^{(n)}(\beta, k, \ell) = \mathrm{tr}_{R}[\widetilde{U}_{k,\ell}^{(n)}\hdots \widetilde{U}_{k,\ell}^{(1)}\,(\varrho_{SL}^{(0)}(\beta,\ell)\otimes \varrho_{R}^{(0)}(\beta,k,\ell) )\,\widetilde{U}_{k,\ell}^{(1) \dagger}\,\hdots \,\widetilde{U}_{k,\ell}^{(n) \dagger}], 
\end{equation}
where $\varrho_{SL}^{(0)}(\beta,\ell):=\tau_S(\beta)\bigotimes_{j=1}^{\ell}\tau_{M_j}(\beta)$, $\varrho_{R}^{(0)}(\beta,k,\ell):=\bigotimes_{j=\ell+1}^{m} \tau_{M_j}$ and $\widetilde{U}_{k,\ell}^{(n)}$ is an arbitrary unitary interaction between $SL$ and $k-\ell$ fresh machines occurring immediately prior to timestep $n$ (see Fig.~\ref{fig:markovianembedding}). Due to the fact that no memory transportation occurs on the $SL$ level throughout the protocol, the full dynamics of the system and memory carriers is captured by the following concatenation of dynamical maps (\emph{c.f.} Eq.~\eqref{eq:appopcpdiv} in contrast):
\begin{equation}\label{eq:cpdivisibleSL}
 \varrho_{SL}^{(n:0)}(\beta,k,\ell) =  \widetilde{\Lambda}_{k,\ell}^{(n:t)} \circ \widetilde{\Lambda}_{k,\ell}^{(t:0)} [\varrho_{SL}^{(0)}(\beta,\ell)],
\end{equation} 
where $\widetilde{\Lambda}_{k,l}^{(n:t)}$ is a CPTP map that acts only upon $SL$ and depends on the unitary operators $\widetilde{U}_{k,l}^{(n)}, \hdots, \widetilde{U}_{k,l}^{(t)}$ and the initial state of the $k-\ell$ fresh machines taking part in each interaction. Thus the dynamics is (operationally) CP-divisible on the level of $SL$, and it is easy to see that it is even Markovian in the stronger sense provided in Refs.~\cite{pollock_operational_2018,Costa2016}.

\begin{figure}[t!]
\centering
\includegraphics[width=\linewidth]{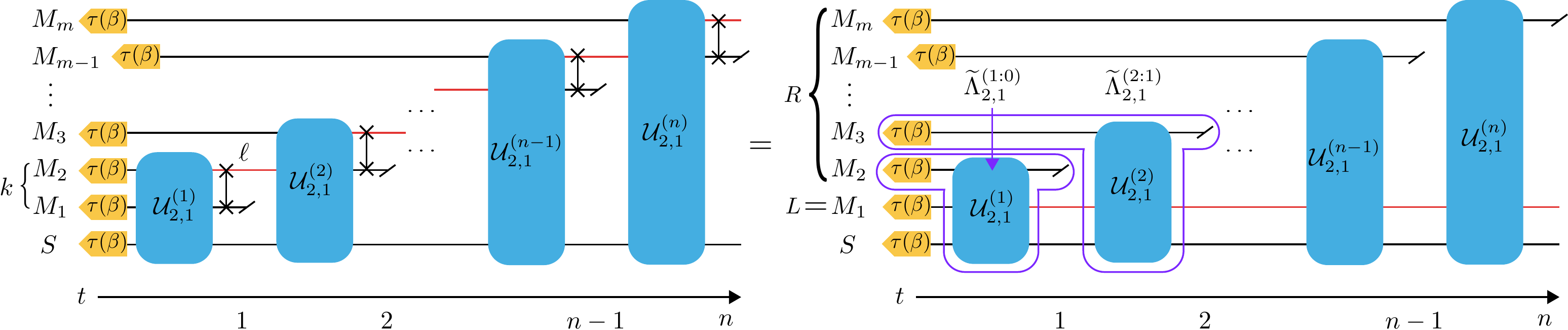}
\caption{ \label{fig:markovianembedding}\textit{Markovian embedding of generalized collision model.} The generalized collision models considered in this work can be embedded as a Markovian process by grouping together the system and memory carriers into a larger target system. On the left, we show that by considering the original circuit shown in Fig.~\ref{fig:collisioncircuit}, plus allowing for a swap interaction between relevant machines, the dynamics can be transformed into the circuit in the right panel. Here, we identify the memory carrier systems as $L$; the entire $SL$ system now interacts with fresh machines between each timestep, which we label $R$. The full dynamics of $SL$ can be described by a Markovian sequence of CPTP maps: here we show only the first two, $\widetilde{\Lambda}^{(1:0)}_{2,1}, \widetilde{\Lambda}^{(2:1)}_{2,1}$ within the purple borders, although the dynamics between any steps can be described similarly. Note that, in contrast to the dynamics of the system itself, the dynamical maps on the level of $SL$ contain the complete description of the process and Eq.~\eqref{eq:cpdivisibleSL} always holds; intuitively, this is because none of the systems carrying memory are artificially ``cut'' by the description of the dynamical map (see the partial traces on the red lines in Fig.~\ref{fig:cpdivisibility} for comparison).}
\end{figure} 

This ``Markovian embedding'' of the non-Markovian dynamics provides an opportunity to investigate the problem at hand with a simplified Markovian dynamics on the larger $SL$ system instead of complicated non-Markovian dynamics that occurs on the level of $S$. In the sense of Ref.~\cite{Campbell2018}, the number of memory carriers $\ell$ corresponds to the memory depth of the dynamics; intuitively, this is the number of additional subsystems that need to be included in the description of the system so that the dynamics is rendered Markovian.

\section{Proof of Theorem 1}\label{app:proof}

Here we prove Theorem~\ref{thm:asymptoticmain}. The proof makes use of the main result of Refs.~\cite{Clivaz_2019L,Clivaz_2019E}, which derive the ultimate cooling bounds for a Markovian protocol. We first embed the non-Markovian dynamics of $S$ as a Markovian one by considering the target system $SL$, before finding the optimally cool $SL$ state in the asymptotic limit, which we denote $\varrho_{SL}^*$. We then combine this result with the fact that there always exists a unitary that can finally be implemented on just $SL$ such that the reduced state of $S$ majorizes all of the possible reduced states of the system ${SL}$, as long as $\varrho_{SL}$ is majorized by $\varrho_{SL}^*$. This implies that the optimal asymptotic system state $\varrho_S^*$ can be calculated from the reduced state of any $\varrho_{SL}$ which has the same eigenvalue spectrum of the asymptotically optimal $SL$ state. 

Before we begin with the proof, we provide a definition of majorization for completeness: 
\begin{defn}
Given a vector of real numbers $\mathbf{a} \in \mathbbm{R}^d$, we denote by $\mathbf{a}^\downarrow$ the vector with the same components but sorted in non-increasing order. Given $\mathbf{a}, \mathbf{b} \in \mathbbm{R}^d$, we say that $\mathbf{a} \succ \mathbf{b}$ ($\mathbf{a}$ majorizes $\mathbf{b}$) iff
\begin{align}
    \sum_i^k a_i^\downarrow \geq \sum_i^k b_i^\downarrow \quad \forall \; k \in \{ 1, \hdots, d\}.
\end{align}
\end{defn}

\textit{Proof of Theorem~\ref{thm:asymptoticmain}}. We first perform a Markovian embedding of the non-Markovian collision model dynamics by considering the evolution of the larger target system $SL$ (with dimension $d_S d_M^\ell$); for a given number $\ell$ of memory carriers, the Markovian embedding corresponds to a memory depth of $\ell$ in the sense of Ref.~\cite{Campbell2018}, which is to say that by including the description of the $\ell$ memory carrying systems with that of the original target system $S$, the dynamics is rendered Markovian. This is because, at each step of the protocol, $SL$ interacts with $k-\ell$ \emph{fresh} machine systems (with total dimension $d_M^{k-\ell}$), which are subsequently discarded and play no further role in the dynamics.

In the Markovian regime, we can use the theorem of universal cooling bound presented in Ref.~\cite{Clivaz_2019L}, which holds for an arbitrary target system interacting with an arbitrary machine, which are initially in a thermal state with inverse temperature ${\beta}$, in the limit of infinite cycles. 

\begin{lem}[Markovian asymptotic cooling limit {[Theorem~1 in Ref.~\cite{Clivaz_2019L}]}] \label{lem:Markov}
For any $d_{\widetilde{S}}$-dimensional system with Hamiltonian $H_{\widetilde{S}}=\sum_{i=0}^{d_{\widetilde{S}} -1}{E}_i\, \ketbra{i}{i}_{\widetilde{S}}$ interacting with
a $d_{\widetilde{M}}$ dimensional machine with Hamiltonian $H_{\widetilde{M}} = \sum_{i}^{d_{\widetilde{M}}-1}\mathcal{E}_i \ket{i}\bra{i}_{\widetilde{M}}$ with $\{ E_i \}, \{ \mathcal{E}_i \}$ sorted in non-decreasing order, in the limit of infinite cycles, 
\begin{itemize}
\item The ground state population of the target system $\widetilde{S}$ is upper bounded by
\begin{equation}
    \widetilde{p}^*(\beta)= (\sum_{n=0}^{d_{\widetilde{S}} -1}e^{-\beta n \widetilde{\mathcal{E}}_{\textup{max}}})^{-1}
\end{equation}
where $\widetilde{\mathcal{E}}_{\textup{max}}$ is the largest energy gap of the machine.
    \item In both coherent and incoherent control scenarios, the vectorized form of eigenvalues of the final state is majorized by that of the following state,
\begin{equation}
   \varrho^*_{\widetilde{S}}(\beta)=\sum_{n=0}^{d_{\widetilde{S}} -1}\frac{e^{-\beta n \widetilde{\mathcal{E}}_{\textup{max}}}}{\mathbbm{Z}_{\tiny{\widetilde{S}}}(\beta,\widetilde{\mathcal{E}}_{\textup{max}})} \, \ketbra{n}{n}_{\widetilde{S}}, 
\end{equation}
if the initial state $\varrho^{(0)}_{\widetilde{S}}(\beta)$ is majorized by $\varrho^*_{\widetilde{S}}(\beta)$.

\item In the coherent control paradigm, the asymptotically optimal state, which is also achievable, is given by $\varrho^*_{\widetilde{S}}(\beta)$.
\end{itemize}
\end{lem}

In view of the fact that the final state $\varrho_{\widetilde{S}}^*(\beta)$ has a unique eigenvalue distribution and is achievable in the infinite-cycle limit, it is possible to investigate this bound on the population of a particular subspace of dimension $d$, rather than just the ground state population. It is straightforward to show that its population is upper bounded by the $d$ largest eigenvalues
\begin{equation}
   \widetilde{p}^*_{\tiny{\mathcal{H}_d}}(\beta)\leq \frac{\sum_{n=0}^{d -1}e^{-\beta n \widetilde{\mathcal{E}}_{\textup{max}}}}{\mathbbm{Z}_{\tiny{\widetilde{S}}}(\beta,\widetilde{\mathcal{E}}_{\textup{max}})}.
   \label{eq:targetsubspace}
\end{equation}

With this knowledge, we are in a suitable position to study the optimal cooling of the non-Markovian collision model protocol in the limit of infinite cycles by employing Lemma~\ref{lem:Markov}. In our case, an arbitrary target system $S$ interacts with $k$ machines at each step and $\ell$ of these carries memory forward to be involved in the next interaction. Thus, the joint system $SL$ corresponds to the target system $\widetilde{S}$ here, which undergoes Markovian dynamics with respect to the $k-\ell$ fresh machines added at each step, which comprise $\widetilde{M}$; hence, $\widetilde{\mathcal{E}}_{\textup{max}}$ is equal to $(k-\ell)\mathcal{E}_{\textup{max}}$. It turns out that the maximum energy gap of the fresh machines and the total dimension and number of the memory carriers play an important role in the ultimate cooling bound.

Using our Markovian embedding of the dynamics and Lemma~\ref{lem:Markov}, we see that in the limit of infinite cycles for any control paradigm, the vector of the eigenvalues of the asymptotic state is majorized by
\begin{equation}\label{eq:appcoolSL}
       \varrho^*_{SL}(\mathcal{E}_\textup{max},\beta,k,\ell)=\sum_{n=0}^{d_{S}d_M^\ell -1}\frac{e^{-\beta n (k-\ell){\mathcal{E}}_{\textup{max}}}}{\tiny{\mathbbm{Z}_{\tiny{SL}}(\beta, (k-\ell)\mathcal{E}_\textup{max})}}\, \ket{n}\bra{n}_{SL} 
\end{equation}
if $\varrho_{SL}^{(0)}(\mathcal{E}_\textup{max},\beta,k,\ell) \prec \varrho_{SL}^{*}(\mathcal{E}_\textup{max},\beta,k,\ell)$ and $\{\ket{n}_{SL}\}$ is the energy eigenbasis with respect to which the energy eigenvalues are sorted in non-decreasing order. So far, we have found the achievable passive state that majorizes all other reachable states of $SL$ via unitary operations on $SLR$. However, this state is not unique as the characterization is based solely on its eigenstate distribution: one can indeed find a whole set of equally cool reachable states, i.e., those for which $\vec{\lambda} [\varrho_{SL}^{*}(\mathcal{E}_\textup{max},\beta,k,\ell)]=\vec{\lambda} [U_{SL}\varrho_{SL}^{*}(\mathcal{E}_\textup{max},\beta,k,\ell)U_{SL}^\dagger]~\forall \; U_{SL}$, where $\vec{\lambda} [\varrho]$ indicates the vectorized form of the eigenvalues of $\varrho$. We now present another lemma which says that from any such state of $SL$, one can reach the optimally cool state of $S$, $\varrho_{S}^{*}$, helping us complete the proof.

\begin{lem}[Reduced state majorization]
\label{lem:maj subsystem}
For any pair states $\varrho_{AB}$ and $\sigma_{AB}$, if $\sigma_{AB}\prec \varrho_{AB}$, there exists a unitary  $U_{AB}^{\textup{opt}}$ on $\mathcal{H}_{AB}$ such that:
\begin{equation}
    \mathrm{tr}_{B}{\left[U_{AB}\sigma_{AB}U_{AB}^{\dagger}\right]}\, \prec \, \mathrm{tr}_{B}{\left[U_{AB}^\textup{opt}\varrho_{AB}{U_{AB}^{\textup{opt} \; \dagger}}\right]}~~~~~~~~~~\forall \; U_{AB}.
\end{equation}
\end{lem}

\begin{proof}

Without loss of generality, we assume that the eigenvalues of both states $\varrho_{AB}$ and $\sigma_{AB}$ are sorted in non-increasing order as follows
\begin{align}
    \textbf{p}_{AB}=\{p_{\alpha}^\downarrow\}_{\alpha=0}^{d_Ad_B-1}\quad\quad \textup{and} \quad \quad
    \textbf{q}_{AB}=\{q_{\alpha}^\downarrow\}_{\alpha=0}^{d_Ad_B-1}.
\end{align}
Based on the sorted eigenvalues, $\sigma_{AB}\prec \varrho_{AB}$ if and only if
\begin{equation}
\label{eq:majorization inequality}
\sum_{\alpha=0}^{k} q_{\alpha}^\downarrow\,\leq\, \sum_{\alpha=0}^{k} p_{\alpha}^{\downarrow} ~~~~~\forall \; k\in \{0,1,\hdots,d_Ad_B-1\}.
\end{equation}
Now we aim to find the reduced state $\sigma_A^{\textup{opt}}$ majorizing all of the achievable reduced states possible to generate by a unitary transformation of $\sigma_{AB}$, which we assume to be diagonal in the orthonormal basis $\{\ket{i j}_{AB}\}$ without loss of generality:
\begin{equation}
    \sigma_{AB}=\sum_{i=0}^{d_A-1}\sum_{j=0}^{d_B-1} q_{ij}\ket{ij}\bra{ij}_{AB}
\end{equation}
One can show that it is possible to obtain $\sigma_A^{\textup{opt}}$ from a bipartite state that is diagonal in the same basis. Then we have,
\begin{equation}
   \widetilde{\sigma}_{AB}= U_{AB}^{\textup{opt}} \sigma_{AB}{U_{AB}^{ \textup{opt} \; \dagger}}=\sum_{i=0}^{d_A-1}\sum_{j=0}^{d_B-1} \widetilde{q}_{ij}\ket{ij}\bra{ij}_{AB},
\end{equation}
where $U_{AB}^{\textup{opt}}$ is simply a permutation matrix that reorders the eigenvalues appropriately. The final reduced state is then given by
\begin{equation}
    \widetilde{ \sigma}_{A}=\sum_{i=0}^{d_A-1}\big(\sum_{j=0}^{d_B-1} \widetilde{q}_{ij}\big)\ket{i}\bra{i}_A. 
\end{equation}

We now need to show that the appropriate unitary maximizes the eigenvalues of the reduced state with respect to eigenvalues of $\sigma_{AB}$. If we rearrange the eigenvalues in such a way that $\widetilde{q}_{ij}=q_\alpha^{\downarrow}$ where $\alpha$ is given by $\alpha:= i\, d_B+j $, we obtain $\sigma_A^{\textup{opt}}$ as the following
\begin{equation}
    \sigma_A^{\textup{opt}}=\sum_{i=0}^{d_A-1}{\eta}^{\downarrow}_i\ket{i}\bra{i}_A=\sum_{i=0}^{d_A-1}\big(\sum_{j=0}^{d_B-1} q^{\downarrow}_{\alpha=id_B+j}\big)\ket{i}\bra{i}_A,
\end{equation}
where, due to the sorting of $\{q_\alpha^{\downarrow}\}$, the eigenvalues of $\sigma_A^{\textup{opt}}$ are sorted in non-decreasing order. The final reduced state satisfies the following condition 
\begin{equation}
\ptr{B}{U_{AB}\sigma_{AB}U_{AB}^{\dagger}}\, \prec \,\ptr{B}{U_{AB}^\textup{opt}\sigma_{AB}{U_{AB}^{\textup{opt} \; \dagger}}}=    \sigma_A^{\textup{opt}}~~~~~~~~~~\forall \; U_{AB}.
\end{equation}

Similarly one can find $\varrho_A^{\textup{opt}}$ by applying a unitary $V_{AB}^{\textup{opt}}$,
\begin{equation}
    \varrho_A^{\textup{opt}}=\sum_{i=0}^{d_A-1}\xi^{\downarrow}_i \ket{i}\bra{i}_A=\sum_{i=0}^{d_A-1}\big(\sum_{j=0}^{d_B-1} {p}^{\downarrow}_{\alpha=id_B+j}\big)\ket{i}\bra{i}_A,
\end{equation}
whose eigenvalues are also in non-increasing order by construction. The final step of the proof is to show that $\sigma_A^{\textup{opt}}\prec \varrho_A^{\textup{opt}}$ whenever $\sigma_{AB}\prec \varrho_{AB}$. This majorization condition can be recast in the form of
\begin{equation}
\sum_{i=0}^{k} \eta^{\downarrow}_i\,\leq\, \sum_{i=0}^{k} \xi^{\downarrow}_i \Rightarrow \sum_{\alpha=0}^{(k+1)d_B-1} q^{\downarrow}_{\alpha}\,\leq\, \sum_{\alpha=0}^{(k+1)d_B-1} p^{\downarrow}_{\alpha}
~~~~~\forall \; k\in \{0,1,\hdots,d_A-1\}.
\label{eq:star majorization}
\end{equation}
Using inequality~\eqref{eq:majorization inequality}, one can easily show that inequality~\eqref{eq:star majorization} always holds, i.e., $\sigma_A^{\textup{opt}}\prec \varrho_A^{\textup{opt}}$, completing the proof.
\end{proof}

In the next step, we aim to maximize the population of the system ground state, i.e., the maximum population of the specific subspace of the $SL$ target given in Eq.~\eqref{eq:targetsubspace}. One must therefore find the target state $\varrho^*_S(\beta,k,\ell)$ that can be achieved from the states $\varrho_{SL}$ with the same eigenvalues as $  \varrho^*_{SL}(\mathcal{E}_\textup{max},\beta,k,\ell)$, since, from Lemma~\ref{lem:maj subsystem} we know that this state majorizes the largest set of states in $S$. In order to do so, we maximize the eigenvalues of $\varrho^*_S(\mathcal{E}_\textup{max},\beta,k,\ell)$ with respect to those of  $\varrho^*_{SL}(\mathcal{E}_\textup{max},\beta,k,\ell)$. One can appropriately sort the eigenvalues of the system $S$ and the memory carrier machines $L$ with the following unitary:
\begin{align}\label{eq:appmaj}
      U_{SL}^{\tiny{\textup{opt}}} \varrho^*_{SL}(\mathcal{E}_\textup{max},\beta,k,\ell) {U_{SL}^{\tiny{\textup{opt}}}}^\dagger &=\sum_{n=0}^{d_{S} -1}\sum_{j=0}^{d_M^\ell -1}\frac{e^{-\beta (n d_M^\ell+j)(k-\ell) {\mathcal{E}}_{\textup{max}}}}{\tiny{\mathbbm{Z}_{\tiny{SL}}(\beta, (k-\ell)\mathcal{E}_\textup{max})}} \, \ket{nj}\bra{nj}_{SL}\nonumber\\
      &=\sum_{n=0}^{d_{S} -1}\sum_{j=0}^{d_M^\ell -1}\frac{e^{-\beta (n d_M^\ell+j)(k-\ell) {\mathcal{E}}_{\textup{max}}}}{\sum_{n'=0}^{d_{S} -1}\sum_{j'=0}^{d_M^\ell -1}e^{-\beta (n' d_M^\ell+j')(k-\ell) {\mathcal{E}}_{\textup{max}}}} \, \ketbra{nj}{nj}_{SL}\nonumber\\
      &=\sum_{n=0}^{d_{S} -1}\frac{e^{-\beta n d_M^\ell(k-\ell) {\mathcal{E}}_{\textup{max}}}}{\sum_{n'=0}^{d_{S} -1}e^{-\beta n' d_M^\ell(k-\ell) {\mathcal{E}}_{\textup{max}}}} \, \ket{n}\bra{n}_{S}\otimes
      \sum_{j=0}^{d_M^\ell -1}\frac{e^{-\beta j(k-\ell) {\mathcal{E}}_{\textup{max}}}}{\sum_{j'=0}^{d_M^\ell -1}e^{-\beta j'(k-\ell) {\mathcal{E}}_{\textup{max}}}} \, \ket{j}\bra{j}_{L}, 
\end{align}

Thus, beginning with the optimally cool $SL$ state in Eq.~\eqref{eq:appcoolSL}, we can reorder the eigenvalues via $U_{SL}^{\tiny{ \textup{opt}}}$ in Eq.~\eqref{eq:appmaj} such that the subsystem $S$ is optimally cool; finally applying Lemma~\ref{lem:maj subsystem} then implies that 
\begin{equation}
    \mathrm{tr}_L[U_{SL}\,   \varrho^*_{SL}(\mathcal{E}_\textup{max},\beta,k,\ell) \, U_{SL}^\dagger ]\prec   \varrho^*_S(\mathcal{E}_\textup{max},\beta,k,\ell)~~~~~~~~\forall \; U_{SL},
\end{equation}
where $\varrho^*_{S}(\mathcal{E}_\textup{max},\beta,k,\ell)$ is indeed given by taking the partial trace over $L$ of $\varrho_{SL}^*(\mathcal{E}_\textup{max},\beta,k,\ell)$:
\begin{align}
       \varrho^*_S(\mathcal{E}_\textup{max},\beta,k,\ell)&=\sum_{n=0}^{d_{S} -1}\frac{\sum_{j=0}^{d_M^\ell -1}e^{-\beta (n d_M^\ell+j)(k-\ell) {\mathcal{E}}_{\textup{max}}}}{\sum_{n'=0}^{d_{S} -1}\sum_{j'=0}^{d_M^\ell -1}e^{-\beta (n' d_M^\ell+j')(k-\ell) {\mathcal{E}}_{\textup{max}}}} \, \ketbra{n}{n}_{S} \nonumber\\
       &=\sum_{n=0}^{d_{S} -1}\frac{e^{-\beta n d_M^\ell(k-\ell) {\mathcal{E}}_{\textup{max}}}}{\sum_{n'=0}^{d_{S} -1}e^{-\beta n' d_M^\ell(k-\ell) {\mathcal{E}}_{\textup{max}}}} \, \ketbra{n}{n}_{S},
\end{align}
thus establishing $\varrho^*_{S}(\mathcal{E}_\textup{max},\beta,k,\ell)$ as the optimal system state in the asymptotic limit. 

In conclusion, from Lemma~\ref{lem:Markov}, we know that the final state of the system $SL$, for any control paradigm in the infinite-cycle limit, is majorized by  $\varrho^*_{{SL}}(\mathcal{E}_\textup{max},\beta,k,\ell)$. Consequently, via Lemma~\ref{lem:maj subsystem}, the final state of $S$ is also majorized by $ \varrho^*_{{S}}(\mathcal{E}_\textup{max},\beta,k,\ell)$. Then, the population of the ground state of the system is upper bounded by the sum of the $d_M^\ell$ largest eigenvalues of $\varrho^*_{SL}(\mathcal{E}_\textup{max},\beta,k,\ell)$, i.e.,  ${p}^*(\mathcal{E}_\textup{max},\beta,k,\ell)= (\sum_{n=0}^{d_{{S}} -1}e^{-\beta n d_M^\ell(k-\ell){\mathcal{E}}_{\textup{max}}})^{-1}$. 

We finally prove that in the coherent scenario, the state  $\varrho^*_{{S}}(\mathcal{E}_\textup{max},\beta,k,\ell)$ is achievable in the limit of infinite cycles. Using Lemma~\ref{lem:Markov}, one can easily show that the final state of $SL$ under optimal coherent operations converges to  $\varrho^*_{{SL}}(\mathcal{E}_\textup{max},\beta,k,\ell)$. To do so, we use the fact that in the coherent scenario, one can apply any unitary operation on the system $SL$ at the final step. We then achieve the desired target state $\varrho^*_{{S}}(\mathcal{E}_\textup{max},\beta,k,\ell)$ via employing the unitary $U_{SL}^{\textup{opt}}$, completing the proof. \qed

\section{Role of system-memory carrier correlations}\label{app:correlations}

We here remark on the correlations that can develop between the target system $S$ and memory carriers $L$ throughout the cooling protocols that have been discussed in the main text.

In particular, we have focused on two procedures. The first strategy optimally cools the joint $SL$ system at each timestep, which does not necessarily ensure that $S$ is locally optimally cool; it is only at the final step that the target system is further cooled by transferring entropy away from it and toward the memory carriers. More precisely, this protocol implements the unitary whose action is defined in Eq.~\eqref{eq:sortsl} at each step (which globally cools $SL$), and only finally implements the unitary that ensures $S$ to be locally cool, whose action is defined in Eq.~\eqref{eq:sorts}. This protocol thus focuses in each step on cooling SL \emph{globally}: in effect, it cools SL with respect to its own (global) energy eigenbasis $(\ket{0}_{SL}, \ket{1}_{SL}, \dots, \ket{n}_{SL})$, with $\ket{i}_{SL}$ denoting the $i^{\text{th}}$ excited state of $SL$ and $n=d_S d_M-1$; as such, we refer to it here as the ``global basis cooling protocol''. 

While the above protocol eventually, i.e., at the last step, optimally cools $S$, it does not necessarily do so at each step. To this end, in the main text we presented a second cooling protocol which is step-wise optimal. Intuitively, this protocol globally cools $SL$ optimally at each step (as does the strategy described above), and, given that optimally cool state, additionally performs a unitary on $SL$ to furthermore optimally cool $S$ \emph{locally} at each step; as such, we refer to this scheme as the ``local basis cooling protocol''. In practical terms, one can view this protocol as cooling $SL$ optimally at each step in the locally-ordered energy eigenbasis $(\ket{00}_{SL}, \ket{01}_{SL}, \dots, \ket{0, d_L-1}_{SL}, \ket{10}_{SL}, \dots, \ket{d_S-1, d_L-1}_{SL})$, where $\ket{i}_S$ is the $i^{\text{th}}$ energy excited state of $S$ and $\ket{j}_L$ the $j^{\text{th}}$ energy excited state of $L$. 

Note that, while related by a permutation of the basis elements, this local energy eigenbasis generally differs from the global energy eigenbasis of $SL$, which does not take local information regarding the energy structure into account. For instance, if the target and memory comprise a qubit each with (respectively) distinct energy gaps, $\{ E_S, E_L\}$ such that $E_L > E_S$ without loss of generality, then cooling with respect to the global basis would order the eigenvalues $\lambda_0 \geq \dots \geq \lambda_3$ into the respective subspaces $\{\ket{00}, \ket{10}, \ket{01}, \ket{11}\}$. However, for $S$ to be optimally cool, the eigenvalues need to be sorted in non-increasing order with respect to $\{\ket{00}, \ket{01}, \ket{10}, \ket{11}\}$; this type of ordering is only achieved at the last timestep via the final unitary in the global cooling protocol, but at every timestep in the local cooling protocol. Discrepancies between the locally-optimal and globally-optimal basis ordering typically become more pronounced as systems become more complex, i.e., multi-partite and high-dimensional, highlighting the necessity for careful accounting. In general, the logic above implies that the local cooling protocol will not reach the coldest $SL$ possible state in particular at each step, but nonetheless always reaches one that is unitarily equivalent to it.

We now analyze the evolution of correlations in $SL$, as measured by the mutual information $I(S:L) := H(\varrho_S) + H(\varrho_L) - H(\varrho_{SL})$, where $H(X) := - \tr{ X \log{(X)}}$ is the von Neumann entropy. In both protocols, the joint state begins as a tensor product and therefore has no correlations. Moreover, the asymptotic state of both protocols is also correlation-free, as was shown in Appendix~\ref{app:proof}. It is of particular interest to note that the asymptotic state of the global protocol always has a product state in its unitary orbit that has the coldest possible state that $S$ can be brought to as its marginal.

Nonetheless, although both protocols start and end with states that are completely decorrelated, correlations do build up for both protocols at finite steps, before decreasing asymptotically as shown in Figure~\ref{fig:correlations}. The finite-time behavior of the correlations generally depends on the full spectrum of $SL$ and the number of steps performed (as does the cooling behavior). In particular---in contrast to the coolness of $S$---the behavior of correlations is non-monotonic, and one cannot even establish a hierarchy between the amount of correlations at any finite time of either protocol. Having presented these initial insights, we leave the full analysis of the role of correlations in quantum cooling as an interesting open avenue for pursuit.

\begin{figure}[t!]
\centering
\includegraphics[width=0.6\linewidth]{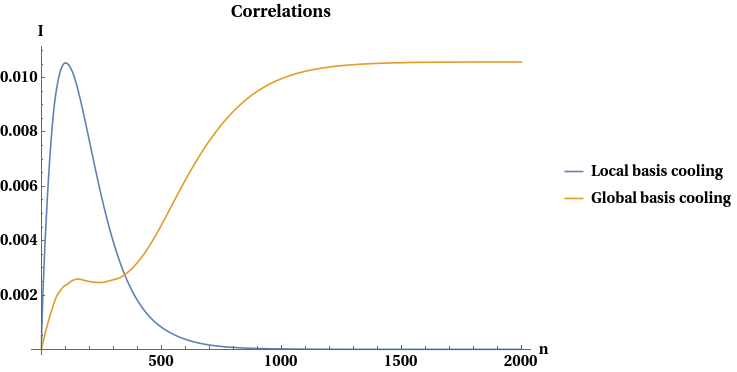}
\caption{ \label{fig:correlations}\textit{Finite-time correlations for global and local cooling protocols.} Here we compare the correlations generated in the global and local cooling protocols, as described in the text. The orange line depicts the protocol that cools $SL$ globally, disregarding the local energy structure, whereas the blue line corresponds to the procedure that ensures $S$ itself is locally optimally cool. Note that, strictly speaking, the global basis cooling protocol includes a final unitary to be implemented at the last timestep, which yields the optimally cool $S$ state; this is not shown above (in order to highlight the distinction between the behavior of correlations for both protocols) but it implies that at any finite time, the joint states achieved by either protocol are unitarily equivalent. In particular, for enough timesteps, the state achieved by the global cooling protocol is arbitrarily close to being unitarily equivalent to a correlation-free product state, which has the coolest possible state of $S$ as its marginal (i.e., the asymptotic state achieved by the local basis cooling protocol). The values used for this simulation are: $d_S=2, d_M=3, k=5, \ell=3, E_1= 1, \mathcal{E}_1=0.5, \mathcal{E}_2=1.2$ and $\beta=0.2$, in natural units where $k_B$ is set to $1$. }
\end{figure}

\section{Step-wise optimal protocol and finite time comparisons}\label{app:finitetime}

Here we provide some analysis on the finite time behavior for the cooling strategies discussed throughout the main text. It is important to note that the finite time properties in general depend upon the details of the full complex energy spectrum of the machines; nonetheless, we have the following observations.

We first detail the step-wise optimal protocol, briefly described in the main text and prove its optimality. 

\begin{defn}[Step-wise optimal cooling unitary]
Given a joint state $\varrho_{SLR}$, let $V^\textup{opt}_{SLR}$ be the unitary that reorders the eigenvalues of $\varrho_{SLR}$ \emph{within each block partitioned by $R$} such that the largest is in the subspace $\ket{000}\bra{000}_{SLR}$, second largest in $\ket{001}\bra{001}_{SLR}$, third largest in $\ket{002}\bra{002}_{SLR}$, $\dots, (d_M^{k-\ell})^{\text{th}}$ largest in $\ket{010}\bra{010}_{SLR}$, and so on until the smallest eigenvalue is in $\ket{d_S-1,d_M^\ell-1,d_M^{k-\ell}-1}$ $\bra{d_S-1,d_M^\ell-1,d_M^{k-\ell}-1}_{SLR}$, i.e., perform
\begin{gather}
\label{eq:opt rhoSLR}
    U^\textup{opt}_{SLR} \varrho_{SLR} {U^{\textup{opt}\; \dagger}_{SLR}} = \sum_{\mu=0}^{d_S-1} \sum_{\nu=0}^{d_M^\ell-1} \sum_{\omega=0}^{d_M^{k-\ell}-1} \lambda^\downarrow_{\mu \cdot d_M^{k}+\nu \cdot d_M^{k-\ell} + \omega } \ket{\mu \nu \omega}\bra{\mu \nu \omega}_{SLR},
\end{gather}
where $\lambda^\downarrow$ denotes the vector of eigenvalues of $\varrho_{SLR}$ labeled in non-increasing order.
\end{defn}

\begin{thm}[Step-wise optimal cooling protocol]
By applying the unitary defined in Eq.~\eqref{eq:opt rhoSLR} at each step, the cooling protocol is step-wise optimal.
\end{thm}

In the Markovian case, the step-wise optimal protocol simply considers all of the eigenvalues of the joint system-machine at each timestep and optimally reorders them such that the system is as cool as possible. However, such a protocol does not ensure step-wise optimality when memory is present: here, not only must we optimally cool the system by rearranging the eigenvalues of the total accessible state at each step, but we must also ensure that this accessible state at each step is as cool as possible given its history. As the only information pertaining to the history is transmitted by the system $SL$, this means that the optimal protocol must at each step optimally cool $S$, and then subject to this constraint, optimally cool the memory carriers $L$ which go on to further cool the system at later times. 

\begin{proof}
We first need to show that $\varrho_S^\textup{opt} := \ptr{LR}{U^\textup{opt}_{SLR} \varrho_{SLR} {U^{\textup{opt}\; \dagger}_{SLR}}}$ obtained from Eq.~\eqref{eq:opt rhoSLR} majorizes the all of the reachable marginal states of $S$; this problem reduces to a constrained rearrangement of the eigenvalues of the entire system, i.e., the eigenvalues are to be arranged optimally with respect to certain eigenspaces. Since majorization theory is independent of the eigenbasis, we choose the energy eigenbasis for simplicity. 

To obtain the eigenspectrum of the system $S$ that majorizes all of the reachable states under unitary transformations on $SLR$, note that the output state of the entire system can be written in the form of 
\begin{gather}
\label{eq:opt rhoSLR tilde}
    \widetilde{U}^\textup{opt}_{SLR} \varrho_{SLR} {\widetilde{U}^{\textup{opt} \; \dagger}_{SLR}} = \sum_{\mu=0}^{d_{SL}-1} \sum_{\eta=0}^{d_M^k-1}  \lambda^\downarrow_{\mu \cdot d_M^{k}+\eta } \ket{\mu \eta}\bra{\mu \eta }_{SLR},
\end{gather}
where $\ket{\mu\, \eta}_{SLR}=\ket{\mu}_{S}\otimes \ket{ \eta}_{LR}$ and $d_{SL} = d_S d_M^\ell$. By the ordering of the eigenvalues that the unitary performs, it is straightforward to see that the $S$ marginal following the optimal transformation majorizes all others in the unitary orbit.

Second, we show that the state of the memory carriers after applying the optimal unitary, i.e., $\varrho_L^\textup{opt} := \ptr{SR}{U^\textup{opt}_{SLR} \varrho_{SLR} {U^{\textup{opt}\; \dagger}_{SLR}}}$, also majorizes all of the reachable states of $L$ given the mentioned majorization condition. We must therefore rearrange the eigenvalues of $\varrho_{SLR}$ within each block corresponding to a fixed $\mu$, i.e., sort $\{\lambda^\downarrow_{\mu \cdot d_M^{k}+\eta }\}_{\eta=0}^{d_M^{k}-1}$ in such a way that the $\nu^{\text{th}}$ largest $d_M^{k-\ell}$ eigenvalues are placed in the $\nu^{\text{th}}$ eigenspace of the system $L$, which gives the state $\varrho_L^\textup{opt}$ that majorizes all of those reachable via unitary transformations on $SLR$. To do so, we rearrange the eigenvalues of the joint $SLR$ system as $\lambda^\downarrow_{\mu \cdot d_M^{k}+\nu \cdot d_M^{k-\ell} + \omega } $, via the unitary transformation defined in Eq.~\eqref{eq:opt rhoSLR}, where $\ket{\mu \nu \omega}_{SLR}=\ket{\mu}_{S}\otimes\ket{\nu}_L\otimes \ket{\omega}_R$. It is clear that the reduced state satisfies the required majorization condition for $\varrho_{L}$, i.e., for all $\mu$, we have
\begin{equation}
   \lambda^\downarrow_{\mu \cdot d_M^{k}+\nu \cdot d_M^{k-\ell} + \omega } \geq  \lambda^\downarrow_{\mu \cdot d_M^{k}+\nu' \cdot d_M^{k-\ell} + \omega' }  ~~~~~~~\text{if} ~~~\nu>\nu'~~~~~~~\forall \; \omega, \omega' \in\{0,1,\hdots,d_M^{k-\ell}-1\},
\end{equation}
where this inequality holds due to the eigenvalue ordering of joint state of $SLR$.

Finally, we show that the output state of the system $SL$ from Eq.~\eqref{eq:opt rhoSLR} majorizes all of those reachable states of $SL$. To do so, must show that $\nu^{\text{th}}$ largest $d_M^{k-\ell}$ eigenvalues of $SLR$ only contribute to the $\nu^{\text{th}}$ eigenvalue of $SL$. This statement follows from
\begin{equation}\label{eq:slmajoriationapp}
   \lambda^\downarrow_{\mu \cdot d_M^{k}+\nu \cdot d_M^{k-\ell} + \omega } \geq  \lambda^\downarrow_{\mu' \cdot d_M^{k}+\nu' \cdot d_M^{k-\ell} + \omega' }  ~~~~\text{if} ~~~\mu \cdot d_M^{k}+\nu \cdot d_M^{k-\ell}>\mu' \cdot d_M^{k}+\nu' \cdot d_M^{k-\ell}~~~~~\forall \; \omega, \omega' \in\{0,1,\hdots,d_M^{k-\ell}-1\}.
\end{equation}

Eq.~\eqref{eq:slmajoriationapp} states that under the protocol considered, one achieves the $SL$ state that majorizes all other reachable states via unitaries on $SLR$. We now need to show that achieving this at every finite timestep $j$ is necessary for subsequent optimal cooling, i.e., that any other protocol is suboptimal. By the stability of majorization under tensor products~\cite{Bondar2003}, we know that $\varrho^{(j) \textup{opt}}_{SL} \otimes \tau_R$, where $\varrho^{(j) \textup{opt}}_{SL} := \ptr{R}{\widetilde{U}^\textup{opt}_{SLR} \sigma_{SLR}^{(j)} {\widetilde{U}^{\textup{opt} \;\dagger}_{SLR}}}$ for any global state $\sigma_{SLR}^{(j)}$, majorizes all of the states $\varrho^{(j)}_{SL} \otimes \tau_R$, where $\varrho^{(j)}_{SL}$ is generated by any other protocol and $\tau_R$ are the thermal bath machines to be added at said timestep. This majorization relation cannot be changed by performing the optimal $SLR$ unitary on $\varrho^{(j) \textup{opt}}_{SL} \otimes \tau_R$ and any other unitary on $\varrho^{(j)}_{SL} \otimes \tau_R$ as the next step of the transformation, with the former therefore yielding $\varrho_{SLR}^{(j+1) \text{opt}}$ and the latter some suboptimal $\varrho_{SLR}^{(j+1)}$. Lastly, invoking the subspace majorization result of Lemma~\ref{lem:maj subsystem}, it follows that $\varrho^{(j+1) \textup{opt}}_{SL} \succ \varrho^{(j+1)}_{SL}$.

Thus, we have shown that at each step of the protocol, we have reached the optimal $SL$ state possible given the history; it is important to note that at this level, the process is Markovian, allowing for an inductive extension of the above argumentation to hold. By further invoking Lemma~\ref{lem:maj subsystem} on the level of $SL$ at each timestep, we yield the optimally cool state of the system $S$, thereby completing the proof.
\end{proof}

\section{Relation to heat-bath algorithmic cooling and state-independent asymptotically optimal protocol}\label{app:hbac}

Here, we propose a general and robust heat-bath algorithmic cooling \textbf{(HBAC)} technique, which we show to be a special case of our generalized collision model, to optimally cool down a target system in the limit of infinite cycles. To obtain the cooling limit most rapidly, in general one must adapt the operations based on the state of $SL$ output by the dynamics at the most recent step. However, via the correspondence between the generalized collision model and HBAC, we can show that not only it is possible to cool down the system by a state-independent, fixed sequence of operations, but also that the protocol converges to the optimally cool state in the asymptotic limit. The result hence draws attention to the fact that in the limit of infinitely many repeated cycles, the dimension of the memory carriers of the protocol (not necessarily knowledge about the state at intermediate times) plays an important role and can already lead to exponential improvement over the Markovian case; in fact, perhaps surprisingly, the role of memory depth is more significant than that of the ability of the agent to implement multi-partite interactions between the system and machines at each step (although, of course, the number of memory carriers is upper bound by how multi-partite the interactions are allowed to be).

Here we will consider the effect of adding compression systems (in the terminology of the HBAC community) or a number of machines that carry memory forward (in the language of our generalized collision model) for a non-adaptive cooling protocol in which a fixed interaction between the target system and a subset of machines at each timestep is repeated infinitely many times. As we have previously, we assume that $k$ machines interact with the system at each step and $\ell$ of them carry memory forward to the next step. This means that $(k-\ell)$ fresh machines and $\ell$ memory carriers participate in the interaction at each timestep. We fix at the outset, for any given choice of these parameters, the dimension of the machines $d_M$, which (along with $k$ and $\ell$) fixes the control complexity in each of the many cases we will look at, and we also fix the temperature at which everything begins, $\beta$.

In Appendix~\ref{app:markovianembedding}, we showed how the dynamics of the system $S$ in the non-Markovian collision model can be described by Markovian dynamics on the larger system $SL$ (with total dimension $d_S d_M^\ell$); in HBAC community, the larger system of such an embedding is known as the computation system, which comprises the original target and what are often referred to as compression / refrigerant systems. In this case, the system $SL$ interact with $k-\ell$ fresh machines (with total dimension $ d_M^{(k-\ell)})$ with maximum energy gap $(k-\ell) \mathcal{E}_{\textup{max}}$; this is known as the reset system, since these are the machines that are discarded after each interaction step, modeling a rethermalization with the external environment. One can decompose the total Hilbert space into the computation part and the reset part $R$, i.e., $\mathcal{H}_{SLR}=\mathcal{H}_{SL} \otimes \mathcal{H}_R$, where here $R$ refers to all of the reset machines comprising the environment. At any timestep, the dynamics of the system $SL$, which arises from unitary evolution on the system $SLR$, is given by (with identity maps implied on the parts of $R$ that do not yet take part in the interaction)
\begin{align}\label{eq:dynamicalmap}
    \varrho^{(n)}_{SL}(\beta,k,\ell) = \Lambda^{(n)}_{k,\ell}(\beta)[\varrho_{SL}^{(n-1)}(\beta,k,\ell)] := \ptr{R}{U^{(n)}_{k,\ell} (\varrho_{SL}^{(n-1)}(\beta,k,\ell) \otimes \varrho_R^{(0)}(\beta,k,\ell)) {U^{(n) \dagger}_{k,\ell}} }.
\end{align}
Note that $\varrho_R^{(0)}(\beta,k,\ell) = \bigotimes_{j=\ell+1}^{k} \tau_{M_j}$ is fixed and the same at each step of the protocol as it refers to the $k-\ell$ fresh machines taken from a thermal bath. In Fig.~\ref{fig:hbac}, we depict the equivalence between the standard HBAC protocol and the generalized collision model formalism.

\begin{figure}[t!]
\centering
\includegraphics[width=\linewidth]{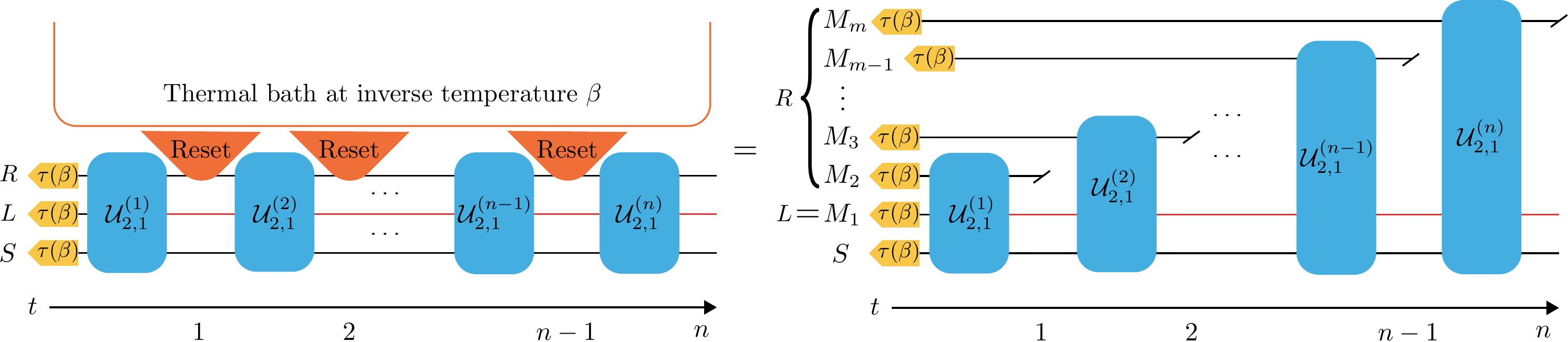}
\caption{ \label{fig:hbac}\textit{Equivalence between HBAC and generalized collision model}. The circuit for a HBAC protocol applied to a quantum system $S$ with one compression system (labeled $L$ to be consistent with our notation) and one reset system $R$. The compression systems store memory of previous interactions (red), whereas the reset ones are assumed to rethermalize with a bath at inverse temperature $\beta$ between each step of the protocol (orange). Noting that the ``reset'' step has the effect of tracing out the system and preparing a fresh one in the thermal state with the same temperature as the bath, it is clear that HBAC is equivalent to generalized collision model we consider (here, $k=2, \ell=1$). Further comparison with Fig.~\ref{fig:markovianembedding} highlights that HBAC need not require the agent to control the compression systems with high fidelity for the entire duration of the protocol: by making appropriate swaps, one only needs to track the $\ell$ compression systems / memory carriers for $\ell$ timesteps.}
\end{figure} 

We now wish to consider a non-adaptive protocol, in which the agent is only allowed to repeatedly apply a fixed unitary operation, i.e., $U_{k,\ell}^{(n)} = U_{k,\ell}~\forall \;n$. The dynamics can then be simplified to
\begin{align}\label{eq:appdynamicalmap}
    \varrho^{(n)}_{SL}(\beta,k,\ell) = \circ^n \Lambda_{k,\ell}(\beta) [\varrho_{SL}^{(0)}(\beta,\ell)],
\end{align}
where $\varrho_{SL}^{(0)}(\beta,\ell)=\tau_S(\beta)\bigotimes_{j=1}^{\ell}\tau_{M_j}(\beta)$ and $\circ^n \Lambda_{k,\ell}(\beta)$ is an $n$-fold concatenation of the dynamical map induced between any pair of timesteps, with $\Lambda_{ k,\ell}(\beta)$ defined such that $\Lambda_{k,\ell}(\beta)[X_{SL}] := \ptr{R}{U_{k,\ell} (X_{SL}\otimes \varrho_R^{(0)}(\beta,k,\ell)) U_{k,\ell}^\dagger}$. This dynamical map is thus independent of the timestep and fully determined by the unitary $U_{k,\ell}$ and the initial state of the fresh machines. In the following, we will show that it is possible to asymptotically reach the ultimate cooling limit via such a non-adaptive protocol. 

\begin{thm}\label{thm:asymptotic}
In the non-adaptive scenario, for a given $d_S$-dimensional system interacting at each step with $k$  $d_M-$dimensional identical machines, with $\ell$ of the machines used at each step carrying the memory forward, in the limit of infinite cycles, it is possible to reach the state $ \varrho^*_S(\beta,k,\ell)$ if the initial state $\varrho_{SL}^{(0)}(\beta,\ell)=\tau_S(\beta)\bigotimes_{j=1}^\ell \tau_{M_j}(\beta)$ is majorized by $ \varrho^*_{SL}(\beta,k,\ell)$. Moreover, it is possible to reach the asymptotic state via a state-independent protocol in which the operation acts on only neighbouring energy levels.
\end{thm}

\begin{proof}
Due to our definition of cooling being based upon majorization, only the eigenvalues of the asymptotic state play a role in determining the fundamental cooling limit. We can therefore restrict our analysis to a specific orthonormal basis, e.g., energy eigenbasis (it is straightforward to generalize the obtained result to an arbitrary orthonormal basis). Here we focus on group of unitary operations that map diagonal density operators of the system $SL$ to diagonal ones. This restriction hence provides us an opportunity to describe the dynamics via stochastic maps that act upon the vector constructed with the eigenvalues of the system and memory carriers.   

The proof takes inspiration from a similar state-independent asymptotically optimal protocol introduced in Ref.~\cite{RaeisiPRL2019}. Here we employ a specific unitary $U$ on the entire $SLR$ system, which can be decomposed as follows
\begin{gather}\label{eq:appfixedunitary}
    U= V \oplus \bar{\openone}
\end{gather}
where $V$ acts unitarily on the Hilbert space $\mathcal{H}_{SL}\otimes \mathcal{H}_G$, in which $\mathcal{H}_G \subset \mathcal{H}_R$ is a subspace spanned by the two eigenstates of the reset systems (fresh machines) that have the maximum energy gap of $(k-\ell)\mathcal{E}_{\textup{max}}$, i.e., $\ket{0}_R$ and  $\ket{d_M^{(k-\ell)}-1}_R$, and $\bar{\openone}$ represents the identity on the subspace $\mathcal{H}_{\widetilde{G}} = \mathcal{H}_R \setminus \mathcal{H}_G$. In the energy eigenbasis, $V$ can also be written in the form of
\begin{equation}\label{eq:nonadaptiveunitary}
    V=\begin{bmatrix} 
   1 &  &  &  & \\
     & \sigma_x &  &  & \\
      &  & \ddots &  & \\
      &  &  & \sigma_x & \\
       &  &  &  &1 \\
    \end{bmatrix}_{2 d_Sd_M^{\ell}\times2 d_Sd_M^{\ell} },
\end{equation}
where $\sigma_x$ is the Pauli X operator. The energy eigenvectors of the Hilbert space $\mathcal{H}_{SL}\otimes \mathcal{H}_G$ are sorted as
\begin{align}
    \ket{2q}_{SLG}&= \ket{2q}_{SLR}=\ket{q}_{SL}\otimes \ket{0}_R\nonumber\\
    \ket{2q+1}_{SLG} &=   \ket{2q+1}_{SLR} =\ket{q}_{SL}\otimes \ket{d_M^{(k-\ell)}-1}_R
\end{align}
with corresponding eigenvalues of $\xi_{2q}={p_q^{(0)}}/{\mathcal{Z}_R}$ and $\xi_{2q+1}={p_q^{(0)}e^{-\beta (k-\ell)\mathcal{E}_{\textup{max}}}}/{\mathcal{Z}_R}$ for $q\in \{0,1,\hdots, d_Sd_M^{\ell}\}$, respectively, where $\mathcal{Z}_R={\mathcal{Z}_{M}(\beta)}^{(k-\ell)}$ is the partition function of the reset system and $\{p_q^{(0)}\}$ are the eigenvalues of the initial state of $SL$.

The unitary $U$ acts to swap every neighboring element on the diagonal part of the global density matrix in the subspace $\mathcal{H}_{SL}\otimes \mathcal{H}_G$ and leave the other elements untouched. We now focus on the transformation of the diagonal elements on the global space $SLR$ under such dynamics. We write the initial state as
\begin{align}
    \varrho_{SLR}^{(0)}&=
    \sum_{r=0}^{2d_S d_M^{\ell}-1}\xi_r^{(0)} \ket{r}\bra{r}_{SLR}
    +\sum_{r=2d_S d_M^{\ell}}^{d_S d_M^{k}-1} \xi_r^{(0)} \ket{r}\bra{r}_{SLR}\nonumber\\
   &=\alpha_{k\ell}\varrho_{SLG}^{(0)}\oplus (1-\alpha_{k\ell}) \varrho_{SL\widetilde{G}}^{(0)}
\end{align}
where $\varrho_{SLG}^{(0)}$ and $\varrho_{SL\widetilde{G}}^{(0)}$ are normalized density matrices and $\alpha_{k\ell}=(1+e^{-\beta (k-\ell )\mathcal{E}_{\textup{max}}})/{\mathcal{Z}_R}$. After applying the unitary $U$, we have
\begin{align}
   \varrho_{SLR}^{(1)}&= U \varrho_{SLR}^{(0)}U^\dagger=\alpha_{k\ell} V\,\varrho_{SLG}^{(0)}\,V^\dagger\oplus (1-\alpha_{k\ell} )\varrho_{SL\widetilde{G}}^{(0)}\nonumber\\
   &=\alpha_{k\ell} \bigg[\xi_0^{(0)} \ket{0}\bra{0}_{SLR}
    + \xi_{d_S d_M^{k}-1}^{(0)} \ket{d_S d_M^{k}-1}\bra{d_S d_M^{k}-1}_{SLR}\nonumber\\
    &+\sum_{r=0}^{d_S d_M^{\ell}-2} \big(\xi_{2r+2}^{(0)} \ket{2r+1}\bra{2r+1}_{SLR}+\xi_{2r+1}^{(0)} \ket{2r+2}\bra{2r+2}_{SLR}\big)\big)\bigg]\oplus (1-\alpha_{k\ell} ) \varrho_{SL\widetilde{G}}^{(0)}
    \label{eq:dynamics rsl1}
\end{align}

It is clear that the output state is also diagonal in the energy eigenbasis. One can easily obtain the reduced state of the system $SL$ after one timestep from Eq.~(\ref{eq:dynamics rsl1}) by taking a partial trace over $R$:
\begin{align}
  \varrho_{SL}^{(1)}&= \ptr{R}{ \varrho_{SLR}^{(1)}} \nonumber\\
&=\alpha_{k\ell} \bigg[\frac{(p_0^{(0)}+ p_1^{(0)})}{1+e^{-\beta (k-\ell)\mathcal{E}_{\textup{max}}}}\ket{0}\bra{0}_{SL}
    + \frac{(p_{d_S d_M^{\ell}-2}^{(0)}+ p_{d_S d_M^{\ell}-1}^{(0)})e^{-\beta (k-\ell )\mathcal{E}_{\textup{max}}}}{1+e^{-\beta (k-\ell )\mathcal{E}_{\textup{max}}}}\ket{d_S d_M^{\ell}-1}\bra{d_S d_M^{\ell}-1}_{SL}\nonumber\\
    &+\sum_{r=1}^{d_S d_M^{\ell}-2} \frac{(p_{r-1}^{(0)}e^{-\beta (k-\ell)\mathcal{E}_{\textup{max}}}+p_{r+1}^{(0)}) }{1+e^{-\beta (k-\ell )\mathcal{E}_{\textup{max}}}} \ket{r}\bra{r}_{SL}\bigg]
    +(1-\alpha_{k\ell} )\sum_{r=0}^{d_S d_M^{\ell}-1} p_{r}^{(0)} \ket{r}\bra{r}_{SL}
    \label{eq:dynamics rs1}
\end{align}
 
Since the output state on $\mathcal{H}_{SLR}$ has a block-diagonal structure with respect to this subspace decomposition, it is locally classical, i.e., has diagonal marginals with respect to the local energy eigenbasis. Therefore, the dynamics of the relevant part of the reduced state can be described in terms of a classical stochastic matrix acting on $SL$ (instead of a CPTP map as would be required if coherences were relevant). In addition, this stochastic matrix is independent of the timestep (since the protocol is non-adaptive) and the $SL$ state at each time. This allows us to describe the evolution of the target system under this protocol via a time-homogeneous Markov process. 

Since the unitary applied does not create coherence in the marginals, it is convenient to introduce a notation for the vectorized form of the diagonal entries of the $SL$ state, i.e., $\textbf{p}_{SL}:=\diag\{p_r\}_{r=0}^{d_S d_M^l - 1}$, where $p_r$ are the eigenvalues of the state $\varrho_{SL}$; since the density matrix is a unit trace positive operator, it follows that the vector $\textbf{p}_{SL}$ has non-negative entries that sum to 1, i.e., it is a probability vector. Then, the state transformation of the system $SL$ between each step of the protocol can be written as
\begin{equation}
    \textbf{p}^{(1)}_{SL}=\big(\alpha_{k\ell}\, \mathbbm{V}\big(\small{(k-\ell )\mathcal{E}_{\textup{max}}}\big) +(1-\alpha_{k\ell} )\mathbbm{1}\big)\textbf{p}^{(0)}_{SL}=:\, \mathbbm{T}\,\textbf{p}^{(0)}_{SL}
\end{equation}
where $\mathbbm{T} $ describes the transition matrix for the Markovian process and the matrix $\mathbbm{V}(\epsilon)$ is given by
\begin{equation}
    \mathbbm{V}(\epsilon)=\small{\frac{1}{1+ e^{-\beta \epsilon}}\begin{bmatrix} 
   1 & 1 &  & \hdots &0 \\
    e^{-\beta \epsilon} & 0& 1 &\hdots  & 0\\
      0& e^{-\beta \epsilon}  &  0&\hdots  & 0\\
     0 & 0 &\hdots  &\ddots & \vdots\\
      0 &  0& \hdots & e^{-\beta \epsilon} & e^{-\beta \epsilon} \\
    \end{bmatrix}_{ d_Sd_M^{\ell}\times d_Sd_M^{\ell} }}
\end{equation}
Since we apply the fixed unitary at each step and the transition matrix is independent of the state of $SL$, the state transformation of $SL$ after n steps can be written as
\begin{equation}
   \textbf{p}^{(n)}_{SL}= \mathbbm{T}^n\,\textbf{p}^{(0)}_{SL}.
\end{equation}
In order to obtain the asymptotic state of the system, we investigate the eigenvalues of the transition matrix given in terms of the two matrices $\mathbbm{V}$ and $\mathbbm{1}$, which allows us to compute the eigenvalues of $\mathbbm{T}$. The eigenvalues of the matrix $\mathbbm{V}$ are presented in Ref.~\cite{RaeisiPRL2019}: $\mathbbm{V} $ has a unique eigenvalue $\nu_0=1$, with the remaining eigenvalues given by 
\begin{equation}
    \nu_q=\frac{2\small{e^{-\frac{\beta}{2} (k-\ell )\mathcal{E}_{\textup{max}}}}\cos{\left(\frac{q\pi}{d_Sd_M^\ell}\right)}}{1+e^{-\beta (k-\ell )\mathcal{E}_{\textup{max}}}} \quad \quad \forall \; q\in\{1,\hdots,d_Sd_M^l-1\}. 
\end{equation}
Since $\mathbbm{1}$ is diagonal with respect to any orthonormal basis and has uniform eigenvalues, it is straightforward to show that the eigenvalues of $\mathbbm{T}$ are obtained by:
\begin{align}
\lambda_q=\alpha_{k\ell} \nu_q + (1-\alpha_{k\ell}).
\end{align}
Thus, $\mathbbm{T}$ also has a unique eigenvalue 1; the eigenvector associated to this value is the steady state solution of dynamics. Moreover, $\mathbbm{T}$ also has the same eigenvectors as $\mathbbm{V}$, since those associated to $\mathbbm{1}$ are trivial. We can then obtain the asymptotic state of the  system $SL$ under a constraint on its initial state, which turns out to only depends on the macroscopic properties of the system and the environment~\cite{RaeisiPRL2019}: %
\begin{equation}
    \textbf{p}^*_{SL}= \underset{n\to \infty}{\lim}\mathbbm{T}^n\,\textbf{p}^{(0)}_{SL}=\left\{\frac{e^{-\beta (k-\ell )q \mathcal{E}_{\textup{max}}}}{\mathbbm{Z}_{SL}(\beta,(k-\ell)\mathcal{E}_{\textup{max}})}\right\}_{q=0}^{d_Sd_M^l-1}.
\end{equation}
This steady state is gives the eigenvalues of the optimally cool achievable state if $\textbf{p}^{(0)}_{SL}\prec \textbf{p}^*_{SL}$. So far, we have shown how one can reach the optimally asymptotic state of $SL$ by employing the fixed unitary in Eq.~\eqref{eq:appfixedunitary} at each iteration. From this asymptotic state, one can easily obtain the coolest achievable reduced state for the system $S$, i.e., $\varrho_S^*$.
\end{proof}

In the non-adaptive scenario, one can further investigate how many repetitions of the cycle are required to achieve the asymptotic state (within a given tolerance). One useful measure for the number of iterations is the mixing time of a Markov process to reach a distance less than $\eta$ to the desired state, i.e., $t_{\textup{mix}} := \min{(n)} : \lvert \textbf{p}^{(n)}_{SL}-\textbf{p}^*_{SL} \rvert \leq \eta$. This mixing time can be upper bounded by a function of difference between the largest and second largest eigenvalues, $\Delta := \lambda_0 - \lambda_1$, as follows
 \begin{align}
  t_{\textup{mix}}(\eta)\leq \frac{1}{\Delta }\log \left(\frac{1}{\eta p^*_{d_Sd_M^l-1}}\right). 
 \end{align}
For the protocol considered above, the spectral gap can be explicitly calculated 
\begin{align}\label{eq:spectralgap}
    \Delta&=\lambda_0-\lambda_1=1-\alpha_{k\ell}\nu_1-(1-\alpha_{k\ell})\nonumber\\
    &=\alpha_{k\ell} \left(1-\frac{2\small{e^{-\frac{\beta}{2} (k-\ell )\mathcal{E}_{\textup{max}}}}\cos{\left(\frac{\pi}{d_Sd_M^\ell}\right)}}{1+e^{-\beta (k-\ell )\mathcal{E}_{\textup{max}}}}\right)\nonumber\\
    &\geq 
    \frac{1+e^{-\beta (k-\ell )\mathcal{E}_{\textup{max}}}}{(\mathcal{Z}_M(\beta))^{k-\ell}}
    \left(\frac{(1-e^{-\frac{\beta}{2} (k-\ell )\mathcal{E}_{\textup{max}}})^2}{1+e^{-\beta (k-\ell )\mathcal{E}_{\textup{max}}}}\right)\nonumber\\
    &=\frac{(1-e^{-\frac{\beta}{2} (k-\ell )\mathcal{E}_{\textup{max}}})^2}{(\mathcal{Z}_M(\beta))^{k-\ell}}.
\end{align}
Then we have
\begin{align}
     t_{\textup{mix}}(\eta)\leq \frac{(\mathcal{Z}_M(\beta))^{k-\ell}}{(1-e^{-\frac{\beta}{2} (k-\ell )\mathcal{E}_{\textup{max}}})^2}\log \left(\frac{1}{\eta p^*_{0} e^{-\beta (k-\ell )(d_Sd_M^l-1)\mathcal{E}_{\textup{max}}}}\right).
\end{align}
This result provides an estimate for the number of iterations of the protocol to reach the optimally cool system.

We now compare the cooling performance between adaptive and non-adaptive strategies for a given choice of memory structure. In the non-adaptive strategy, the rate of cooling is determined completely by the spectral gap $\Delta$ in Eq.~\eqref{eq:spectralgap}, as the same dynamics is repeated at each step. In the adaptive scenario, this is no longer the case and a single parameter does not dictate the rate of convergence to the asymptotic state. Instead, in general, the cooling rate depends upon the entire energy structure of the system and all machines, making a closed form expression difficult to derive. Nonetheless, we can describe the solution to the problem of reaching a step-wise provably optimal system state at finite times as a protocol, as done in the main text. This protocol converges to the same asymptotic value as the non-adaptive case, but offers a finite time advantage, as shown in Fig.~\ref{fig:adaptive}.

\begin{figure}[h!]
\centering
\includegraphics[width=0.6\linewidth]{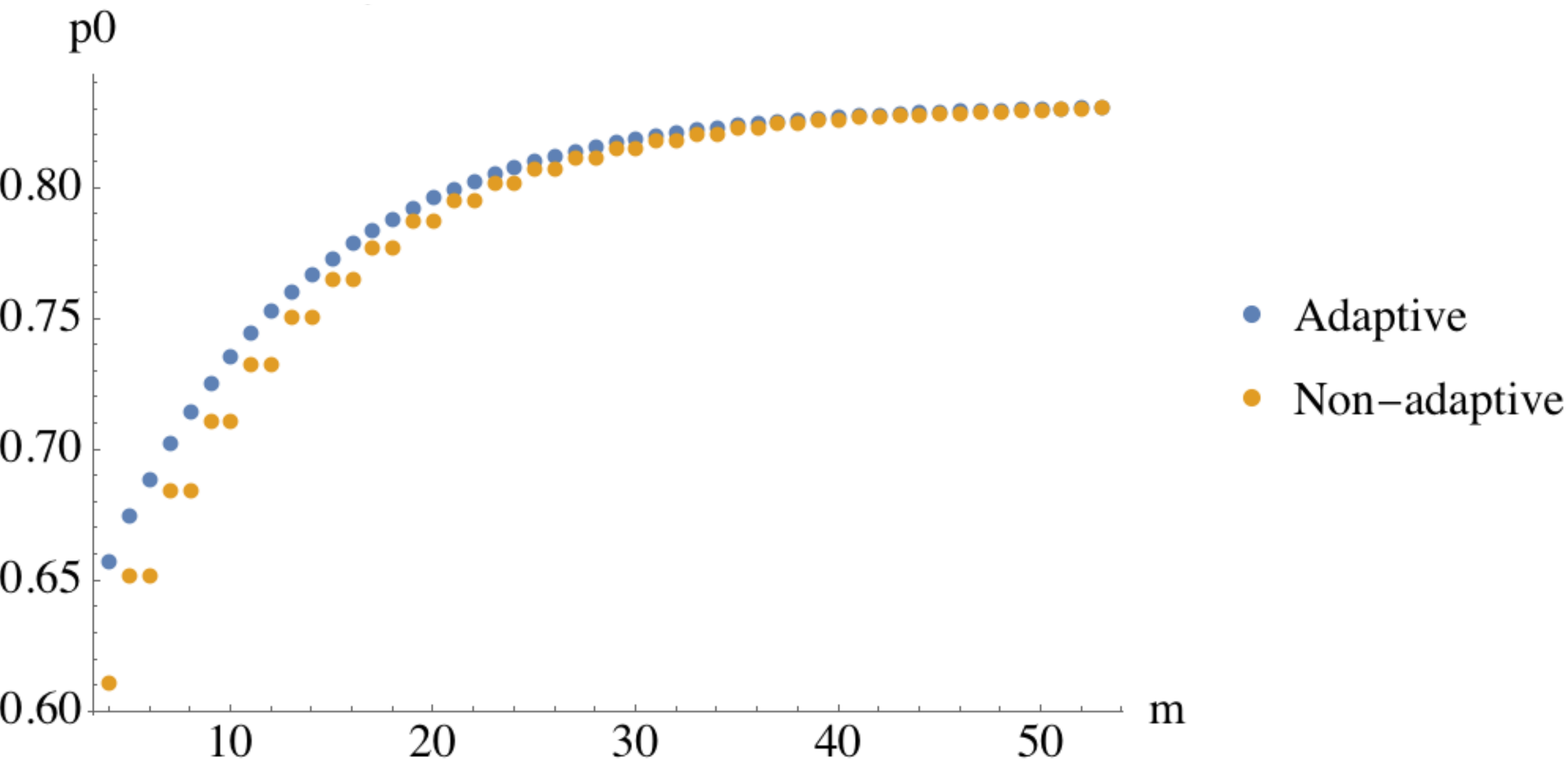}
\caption{ \label{fig:adaptive}\textit{Finite-time advantage with adaptive protocols.} Here we compare the adaptive and non-adaptive protocols for a qubit system (with $E_\textup{max} = 1$) interacting with qubit machines (with $\mathcal{E}_{\textup{max}}=2$) with initial temperature $\beta=0.2$ and memory structure given by $k=3, \ell=2$. In the adaptive scenario, we make use of the step-wise optimal protocol described in the main text; in the non-adaptive, the unitary in Eq.~\eqref{eq:nonadaptiveunitary} is repeatedly implemented. We see that, although both scenarios asymptotically converge to the same ground state population, the adaptive protocol outperforms the non-adaptive one at finite times. This behavior is more pronounced for larger dimensions. On the other hand, the non-adaptive, state-independent protocol is more robust and as it rapidly converges to the asymptotically-optimal state, therefore better suited to practical implementations. Note that, in this case, every second step of the non-adaptive protocol does not act to cool the system.}
\end{figure} 

We lastly compare various memory structures (i.e., values of $k$, $\ell$) with respect to the optimal adaptive protocol. In order to do so in a meaningful way, we compute the ground state population of the system for after a fixed number $m = k + (n-1)(k-\ell)$ machines have been exhausted. If one were to compare the ground state populations after $n$ unitaries had been implemented, for various $k$, $\ell$, one would be making an unfair comparison with respect to the total resources at hand; e.g., after three unitaries with $k=4, \ell=3$ the experimenter has used six machines, whereas for $k=3,\ell=0$, they have used nine. Comparing various scenarios at fixed values of $m$ provides insight into how cool the system can be prepared after all constituents of a finite sized environment are used up for the given memory structure. This change of perspective comes at the cost of the fact that the number of physical unitaries $n$ needed to be implemented in order to exhaust the resources (quantified by $m$) now varies; e.g., to use six machines with $k=4, \ell=3$ takes three unitaries, whereas with $k=1, \ell=0$ it takes six. Lastly note that not all values of $k, \ell$ are valid for a given $m$, due to the restriction that $n$ must be an integer. The short-term behavior is displayed in Fig.~\ref{fig:finitecomp}.

\begin{figure}[h!]
\centering
\includegraphics[width=0.6\linewidth]{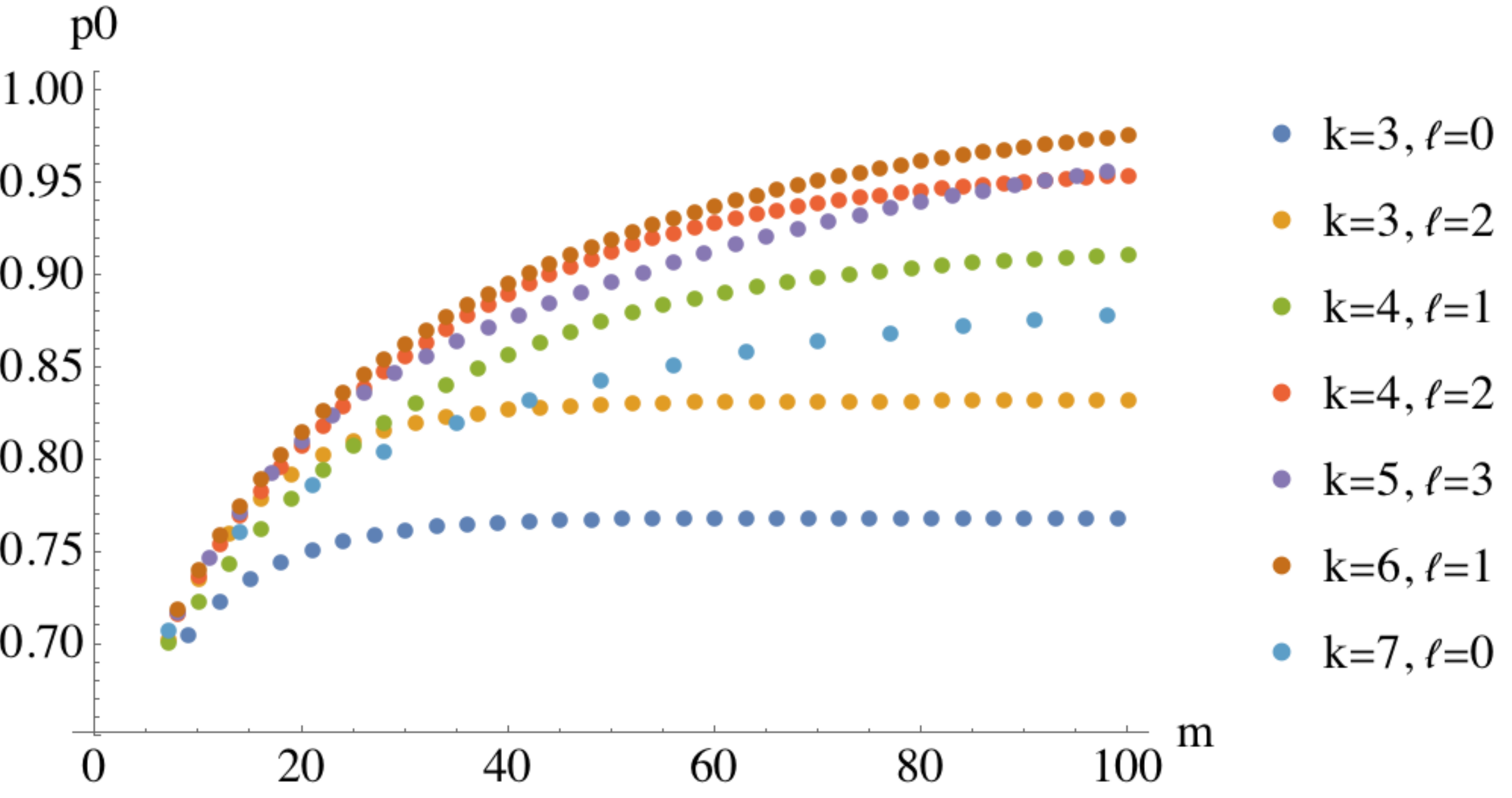}
\caption{ \label{fig:finitecomp}\textit{Finite-time behavior of step-wise optimal protocol.} Here we plot the ground state population after $m$ machines have been completely exhausted, for various values of $k,\ell$ as described in the text (we use a qubit system and qubit machines, with $\beta=0.2, E_{\textup{max}} = 1$ and $\mathcal{E}_{\textup{max}}=2$). The finite time behavior generally depends upon the full structure of the energy spectra, but already here we see some interesting effect. For instance, note that at $m=7$ (first point shown), the $k=7,\ell=0$ case provides the best possible cooling, as it allows for a full 8-partite unitary between the system and seven machines. However, for more applications of such a unitary with in the memoryless scenario, the performance can be worse than other cases with more local interactions (smaller $k$) and longer memory (larger $\ell$).}
\end{figure} 

\newpage


%

\end{document}